\documentclass[11pt]{article}

\usepackage[a4paper,margin=1.05in]{geometry}
\usepackage[T1]{fontenc}
\usepackage{lmodern}
\usepackage{microtype}
\usepackage{amsmath,amssymb,amsthm,mathtools}
\usepackage{enumitem}
\usepackage[hidelinks]{hyperref}
\hypersetup{
  pdftitle={Free-Energy Asymptotics of the Two-Dimensional Quantum Heisenberg Ferromagnet},
  pdfauthor={Andreas Klippel},
  pdfsubject={Low-temperature free energy, collision traces, and magnons in the two-dimensional quantum Heisenberg ferromagnet}
}

\allowdisplaybreaks
\numberwithin{equation}{section}

\newtheorem{theorem}{Theorem}[section]
\newtheorem{proposition}[theorem]{Proposition}
\newtheorem{lemma}[theorem]{Lemma}
\newtheorem{corollary}[theorem]{Corollary}
\theoremstyle{remark}

\newcommand{\N}{\mathbb N}
\newcommand{\Z}{\mathbb Z}
\newcommand{\R}{\mathbb R}

\newcommand{\Tr}{\operatorname{Tr}}
\newcommand{\1}{\mathbf 1}

\newcommand{\cE}{\mathcal E}
\newcommand{\cM}{\mathcal M}
\newcommand{\Sp}{\mathbf S}
\newcommand{\spec}{\operatorname{spec}}
\newcommand{\gap}{\operatorname{gap}}
\newcommand{\Sym}{\operatorname{Sym}}
\newcommand{\dist}{\operatorname{dist}}
\newcommand{\HW}{\mathrm{HW}}

\title{Free-Energy Asymptotics of the Two-Dimensional Quantum Heisenberg Ferromagnet}
\author{Andreas Klippel\thanks{ \href{mailto:anklippe@uni-mainz.de}{anklippe@uni-mainz.de}}\\[0.35em]
\small Institut für Mathematik, Johannes Gutenberg-Universität Mainz,\\[-0.1em]
\small Staudingerweg 9, 55128 Mainz}
\date{\today}

\begin{document}
\maketitle

\begin{abstract}
We prove the leading low-temperature free-energy asymptotics of the
nearest-neighbour quantum Heisenberg ferromagnet on \(\mathbb Z^2\).  For
every fixed \(S\in\{\frac12,1,\frac32,\ldots\}\), the free energy per site
satisfies
\[
 \lim_{\beta\to\infty}\beta^2S f_2(\beta,S)=-\frac{\pi}{24}.
\]
This formula was predicted by Takahashi's modified spin-wave theory
\cite{Takahashi1986}.  Napiórkowski and Seiringer proved the corresponding
upper bound but left the matching lower bound in two dimensions open
\cite{NapiorkowskiSeiringer2021}, a gap later highlighted again by
Seiringer \cite{SeiringerOberwolfach2025}.  We prove this missing bound
through a new comparison between the magnon exclusion dynamics and free
bosons, based on extending wave functions to collision configurations with
controlled kinetic cost.
\end{abstract}



\section{Introduction}
\label{sec:introduction}

The isotropic quantum Heisenberg ferromagnet is one of the standard models
of quantum magnetism.  Its local interaction favours parallel spins, and
its lowest excitations are collective waves of the spin field.  The
spin-wave description goes back to Bloch \cite{Bloch1930}; one quantum of
such a mode is called a \emph{magnon}.
The Holstein--Primakoff representation describes spin deviations by
bosons with an on-site occupation constraint
\cite{HolsteinPrimakoff1940}.  Spin-wave theory replaces the interacting
quantum system, to leading order at low temperature, by an ideal Bose gas
whose one-particle energy behaves as \(|p|^2\) near zero momentum.  This
picture underlies Dyson's low-temperature expansion of the ferromagnet
\cite{Dyson1956,Dyson1956Thermodynamics}.
Takahashi's modified spin-wave theory predicted the leading
low-temperature free energy in dimensions one and two
\cite{Takahashi1986}.

The replacement is far from automatic.  In the exact spin system, the
number of spin deviations at one site is bounded, neighbouring deviations
interact, and global spin rotations produce many zero-energy states.  The
leading free energy therefore depends not only on isolated low-lying
eigenvalues, but on the total number of many-body states whose Boltzmann
weights contribute to \(\Tr e^{-\beta H}\).

Two dimensions make this question especially sharp.  The Mermin--Wagner
theorem excludes spontaneous magnetic order at every positive temperature
for the short-range isotropic model \cite{MerminWagner1966}.  It does not,
however, determine the leading free energy.  The natural question is
whether treating magnons as independent bosons still gives the correct
leading free energy even though spontaneous magnetization is absent.  In
other words, can spin waves control the leading thermodynamics without
long-range order?

In three dimensions, a short account of the spin-\(1/2\) case appeared in
\cite{CorreggiGiulianiSeiringer2014}, while the spin-wave free-energy
asymptotics for every fixed spin were proved in
\cite{CorreggiGiulianiSeiringer2015}; the same argument extends to
\(d\geq3\) \cite{SeiringerOberwolfach2025}.  This followed earlier work in the
large-spin regime \cite{CorreggiGiuliani2012}.  The first-order interaction
correction in that regime was subsequently obtained, as an upper bound, in
both two and three dimensions \cite{Benedikter2017}.  Earlier non-sharp
pressure bounds include \cite{ConlonSolovej1991,Toth1993}.  Napiórkowski and
Seiringer subsequently established the corresponding result for the
one-dimensional chain \cite{NapiorkowskiSeiringer2021}.  In two dimensions
they proved the predicted upper bound but left the matching lower bound
open \cite{NapiorkowskiSeiringer2021}.  Seiringer later identified this as
the remaining open case \cite{SeiringerOberwolfach2025}.

To state the result, let \(\beta>0\) denote inverse temperature and let
\(\Lambda\) be a finite subset of \(\Z^2\).  At each \(x\in\Lambda\), place
the spin-\(S\) irreducible representation of
\(\mathrm{SU}(2)\), with
generators \(\Sp_x=(S_x^1,S_x^2,S_x^3)\).  We consider the free-boundary
nearest-neighbour Hamiltonian
\begin{equation}
 H_\Lambda^{(S)}
 :=\sum_{\substack{\langle x,y\rangle\\x,y\in\Lambda}}
 \bigl(S^2-\Sp_x\mathbin\cdot\Sp_y\bigr),
 \label{eq:intro-hamiltonian}
\end{equation}
where each unoriented nearest-neighbour bond is counted once.  The additive
constant normalizes the ground-state energy to zero.  Write
\begin{equation}
 f_2(\beta,S):=-\lim_{\Lambda\uparrow\Z^2}
 \frac1{\beta|\Lambda|}\log\Tr e^{-\beta H_\Lambda^{(S)}}
 \label{eq:intro-free-energy},
\end{equation}
for the thermodynamic free energy.  The limit may equivalently be taken
along free square boxes.  For completeness, let \(H_0\) be the sum of the
Hamiltonians of disjoint square boxes and let \(B\) be the sum of the bonds
joining them.  Since \(B\geq0\) and \(\|B\|\leq S(2S+1)\) times the
number of joining bonds, the min--max principle gives
\[
 e^{-\beta\|B\|}\Tr e^{-\beta H_0}
 \leq \Tr e^{-\beta(H_0+B)}
 \leq \Tr e^{-\beta H_0}.
\]
Here positivity follows from
\(2\Sp_x\cdot\Sp_y=(\Sp_x+\Sp_y)^2-2S(S+1)\) and the bound of the
two-site total spin by \(2S\); the same identity gives the stated norm.
Tiling a large square by boxes of side \(L\), with a remainder strip of
width less than \(L\), shows that the upper and lower limiting pressures
differ by at most \(C_S\beta/L\).  Letting \(L\to\infty\) proves the
existence of the limit.  We write
\(p_2(\beta,S):=-\beta f_2(\beta,S)\) for the logarithmic pressure per
site.

The known and missing directions can now be stated without ambiguity.
Napiórkowski and Seiringer proved the first inequality below; the second
was open:
\begin{align}
 \limsup_{\beta\to\infty}\beta^2S f_2(\beta,S)
 &\leq-\frac\pi{24},
 \label{eq:intro-known-direction}\\
 \liminf_{\beta\to\infty}\beta^2S f_2(\beta,S)
 &\geq-\frac\pi{24}.
 \label{eq:intro-missing-direction}
\end{align}
Since \(p_2=-\beta f_2\), the second inequality is the corresponding upper
bound on the pressure.  It is the new estimate proved in this paper.

\begin{theorem}[Two-dimensional spin-wave free energy]
\label{thm:main-free-energy}
For every fixed \(S\in\{\frac12,1,\frac32,\ldots\}\),
\begin{equation}
 \lim_{\beta\to\infty}\beta^2S f_2(\beta,S)
 =-\frac\pi{24}.
 \label{eq:main-limit}
\end{equation}
Equivalently,
\[
 f_2(\beta,S)=-\frac{\pi}{24S\beta^2}+o(\beta^{-2}).
\]
\end{theorem}

With \(p\) denoting momentum, the constant is the free two-dimensional Bose
integral
\begin{equation}
 \frac1{(2\pi)^2}\int_{\R^2}
 \log\bigl(1-e^{-|p|^2}\bigr)\,dp=-\frac\pi{24}.
 \label{eq:bose-integral},
\end{equation}
Thus magnon interactions and the on-site occupation bound do not change
the leading pressure.  Theorem~\ref{thm:main-free-energy} does not assert
magnetic order or a Gibbs state with spontaneous magnetization.  It shows
that the summed Boltzmann weights of the low-energy states agree, to
leading order, with the free-spin-wave prediction.

The proof also gives the quantitative remainder stated in
Corollary~\ref{cor:quantitative-free-energy}.  Its exponent is not expected
to be optimal; it records the finite boxes and energy cutoffs used in the
proof.

The central new estimate of the paper concerns the on-site occupation
constraint.
For spin \(1/2\), a sector with \(k\) reversed spins is an exclusion
process: the particles hop but cannot occupy the same site.  Free magnons,
by contrast, are symmetric bosons and may coincide.  We extend a function
from distinct configurations to collisions by minimizing the free kinetic
energy.  The added collision energy is then bounded by the \emph{square}
of the exclusion energy, up to a logarithmic loss in the box size.
Exchange symmetry supplies the additional power of the energy, while an
electrical-network estimate controls configurations containing occupied
neighbouring sites.  This quadratic collision estimate is the main
analytic innovation of the paper.

This mechanism differs from the arguments used in the other dimensions.
In three and higher dimensions the interaction is controlled through an
estimate on the two-particle density
\cite{CorreggiGiulianiSeiringer2015}.  The one-dimensional proof uses
estimates that depend essentially on the order of sites along a chain
\cite{NapiorkowskiSeiringer2021}.  Our estimate controls pair collisions
and several simultaneous collisions at once, uniformly over the particle
numbers and energies retained in the finite-volume trace.

After discarding a quantitatively negligible high-energy part, harmonic
extension and the min--max principle compare the remaining eigenvalues
with those of free symmetric bosons.  The constant one-particle
eigenfunction has energy zero and accounts for the degeneracy under global
spin rotations; the positive eigenvalues give the coefficient \(\pi/24\).

For general fixed \(S\), we replace each spin-\(S\) degree of freedom by
\(r=2S\) auxiliary spin-\(1/2\) variables, called fibres, and couple every
fibre above \(x\) to every fibre above a neighbouring site \(y\).  The
subspace symmetric under permutations of the fibres at each site is the
original spin-\(S\) model.  The component constant in the fibre label
reproduces the physical magnon, whereas every orthogonal component has
energy bounded below by a positive constant independent of the box size.
The collision estimate remains valid with an \(S\)-dependent constant,
which is why the theorem holds for every fixed spin; no uniformity as
\(S\to\infty\) is claimed.

\section{Magnons, exclusion, and the collision correction}
\label{sec:magnons-collision}

To bound the pressure, we compare the interacting \(k\)-particle spectrum
with the spectrum of free symmetric bosons.  For spin \(1/2\), fixing the
number \(k\) of reversed spins turns twice the Hamiltonian into a
\(k\)-particle exclusion Laplacian: particles move between neighbouring
sites but cannot coincide \cite{Thomas1980}.  The main object of this
section is \(N_{L,k}\), the energy added when a function on distinct
configurations is extended to collision configurations with minimal free
kinetic energy.  We first describe the spin decomposition and then
construct this operator.

Throughout the proof, \(\N=\{1,2,\ldots\}\), and \(C,c>0\) denote
constants that may change from line to line; a subscript \(r\) allows
dependence on \(r\).  Unlabelled \(\ell^2\)-inner products and norms use
counting measure, \(\1_A\) denotes the indicator of \(A\), and \(I\) denotes
the identity operator.  We write \(A\asymp B\) when
\(cB\leq A\leq CB\), with constants independent of the asymptotic
parameter under consideration.

\subsection{Highest-weight sectors and exclusion}
\label{subsec:highest-weight}

Let \(L\geq3\) be an integer and let
\(B_L=\{1,\ldots,L\}^2\) with free boundary, or
\(B_L=(\Z/L\Z)^2\) with periodic boundary, and put \(V=L^2\).  Write
\(x\sim y\) for nearest neighbours in the corresponding graph.  For
\(0\leq k\leq V\), set
\[
 \Omega_{L,k}:=\{A\subset B_L:|A|=k\}.
\]
Under the identification of \(A\) with the spin configuration having down
spins exactly at the sites of \(A\), twice the spin-\(1/2\) Hamiltonian is
the following operator on \(\ell^2(\Omega_{L,k})\):
\begin{equation}
 (H_{L,k}f)(A)
 :=\sum_{\substack{x\in A,\ y\notin A\\x\sim y}}
 \bigl(f(A)-f(A\setminus\{x\}\cup\{y\})\bigr).
 \label{eq:exclusion-laplacian}
\end{equation}
This is the \(k\)-particle exclusion Laplacian.  If \(t=\beta/2\), the
physical Gibbs factor in this sector is
\(e^{-tH_{L,k}}\).

Rotation invariance gives a more economical decomposition than the raw
down-spin sectors.  Put
\[
 S^3_{\mathrm{tot}}:=\sum_{x\in B_L}S_x^3,\qquad
 S^+_{\mathrm{tot}}:=\sum_{x\in B_L}(S_x^1+\mathrm iS_x^2).
\]
Let \(\cM_{L,k}\) be the subspace on which
\(S^3_{\mathrm{tot}}=V/2-k\) and \(S^+_{\mathrm{tot}}=0\).  It contains
one highest-weight vector for each copy of the total-spin representation
of spin \(V/2-k\), and is therefore called its highest-weight multiplicity
space.  Write \(H_{L,k}^{\HW}\) for the restriction of \(H_{L,k}\) to
this space.
Then
\begin{equation}
 \Tr e^{-\beta H_{B_L}^{(1/2)}}
 =\sum_{k=0}^{\lfloor V/2\rfloor}
 (V-2k+1)\Tr_{\cM_{L,k}}e^{-tH_{L,k}^{\HW}}.
 \label{eq:su2-trace-decomposition}
\end{equation}
Only the elementary bound \(V-2k+1\leq V+1\) will be used.  We never
require harmonic extension to preserve \(\cM_{L,k}\).

It is convenient to use ordered coordinates before imposing symmetry.  Let
\(X_{L,k}=B_L^k\), let \(\ell^2_{\Sym}(X_{L,k})\) be the permutation-invariant
subspace, and define the free product Laplacian
\begin{equation}
 L_{0,L,k}:=\sum_{i=1}^k(-\Delta_i),
 \qquad
 (-\Delta u)(x)=\sum_{y\sim x}\bigl(u(x)-u(y)\bigr).
 \label{eq:free-product-laplacian}
\end{equation}
Here \(-\Delta_i\) acts as the graph Laplacian in coordinate \(i\) and as
the identity in the remaining coordinates.
For the free square we also denote this graph Laplacian by
\(-\Delta_N\).
The distinct-coordinate set
\[
 D_{L,k}:=\{(x_1,\ldots,x_k):x_i\ne x_j\ \text{for }i\ne j\}
\]
is related to \(\ell^2(\Omega_{L,k})\) by the unitary map
\[
 U:\ell^2(\Omega_{L,k})\longrightarrow\ell^2_{\Sym}(D_{L,k}),
 \qquad
 (Ug)(x_1,\ldots,x_k)
 =(k!)^{-1/2}g(\{x_1,\ldots,x_k\}).
\]
Under this identification the moves staying inside \(D_{L,k}\) give
\(H_{L,k}\).

\subsection{Harmonic extension to free bosons}
\label{subsec:harmonic-extension}

For \(k=1\) there is no collision set and we put \(N_{L,1}=0\).  Let
\(k\geq2\) below and set \(C_{L,k}=X_{L,k}\setminus D_{L,k}\).  With
respect to the orthogonal decomposition into distinct and colliding
configurations, the symmetric product Laplacian has the block form
\begin{equation}
 L_{0,L,k}=
 \begin{pmatrix}
  A&B\\ B^*&C_0
 \end{pmatrix}.
 \label{eq:block-product-laplacian}
\end{equation}
The block \(B\) couples distinct and colliding configurations, and
\(C_0\) is the compression to the collision set with Dirichlet boundary
on \(D_{L,k}\).  Moreover, \(A=H_{L,k}+R_{L,k}\), where \(R_{L,k}\) is
multiplication by the
number of free product jumps from a distinct configuration into the
collision set.  For \(k\leq V/2\), the Dirichlet collision block \(C_0\) is
strictly positive.  Indeed, the product graph \(B_L^k\) is connected and
\(D_{L,k}\) is nonempty for \(k\leq V\).  Starting from any colliding
configuration, take a product-graph path to a fixed element of
\(D_{L,k}\).  The initial segment up to the first entrance into
\(D_{L,k}\) shows that every connected component of \(C_{L,k}\) has an
edge to the Dirichlet boundary.  The quadratic form of \(C_0\) can
therefore vanish only on the zero function.  Since the space is finite
dimensional, \(C_0>0\).

Given \(f\) on \(D_{L,k}\), let \(E_{L,k}f\) be the extension that
minimizes the product energy among all extensions of \(f\).  In the block
decomposition above it is
\begin{equation}
 E_{L,k}f:=\binom{f}{-C_0^{-1}B^*f}.
 \label{eq:harmonic-extension}
\end{equation}
The Schur complement yields
\begin{align}
 \|E_{L,k}f\|&\geq\|f\|,
 \label{eq:extension-norm}\\
 \langle E_{L,k}f,L_{0,L,k}E_{L,k}f\rangle
 &=\langle f,(H_{L,k}+N_{L,k})f\rangle,
 \label{eq:extension-energy}
\end{align}
where
\begin{equation}
 N_{L,k}:=R_{L,k}-BC_0^{-1}B^*.
 \label{eq:normal-correction}
\end{equation}
Applying the minimizing property first to \(E_{L,k}f\) and then to an
extension supported on \(D_{L,k}\) gives
\begin{equation}
 0\leq N_{L,k}\leq R_{L,k}.
 \label{eq:N-between-zero-R}
\end{equation}
The constant function extends to a constant function, so both \(H_{L,k}\)
and \(N_{L,k}\) vanish on constants.  Theorem~\ref{thm:normal-trace}
quantifies their relation on the orthogonal complement.

\section{A quadratic trace estimate on collision diagonals}
\label{sec:collision-trace}

The main object of this section is the kinetic energy carried by product
edges which touch a collision diagonal.  We first show that exchange
symmetry makes this energy second order in the free product Laplacian.  We
then bound the collision multiplier \(R_{L,k}\) by the exclusion Dirichlet
form, using effective resistance, and obtain an estimate involving only the
exclusion operator.

\begin{theorem}[Quadratic collision-trace estimate]
\label{thm:normal-trace}
There is a universal constant \(C_{\mathrm{tr}}<\infty\) such that, for
every integer \(L\geq3\), for either choice of boundary conditions
(periodic or free), and for every integer
\(1\leq k\leq\lfloor V/2\rfloor\),
\begin{equation}
 0\leq N_{L,k}\leq
 C_{\mathrm{tr}}k^6\log^2(2L)\,H_{L,k}^2
 \label{eq:intro-normal-trace}
\end{equation}
on the Hilbert space of permutation-invariant functions of \(k\) distinct
particle coordinates.
\end{theorem}

The square of \(H_{L,k}\) is essential.  On the low-energy subspace used
in the pressure estimate it makes the collision correction small relative
to \(H_{L,k}\), even though both the box and the number of magnons grow
with the inverse temperature.  The logarithm is the borderline
two-dimensional resistance cost.

\subsection{Exchange-symmetric cancellation}
\label{subsec:exchange-cancellation}

We begin with periodic boundary conditions.  A product edge is an edge of
\(B_L^k\) obtained by moving one coordinate by one nearest-neighbour step.
Call it bad if at least one endpoint belongs to \(C_{L,k}\).  For either
orientation \(e=(u,v)\), set \(\nabla_eF:=F(v)-F(u)\); its squared modulus
is independent of the chosen orientation.  Write
\begin{equation}
 q_{\mathrm{bad}}(F)
 :=\sum_{e\,\mathrm{bad}}|\nabla_eF|^2.
 \label{eq:bad-edge-form}
\end{equation}
All edge sums in this section are over unoriented product edges, counted
once.  Thus the bad edges are exactly those product edges which do not
have both endpoints in \(D_{L,k}\).
For \(1\leq i<j\leq k\), let \(\mathcal B_{ij}\) contain the product edges
based on the diagonal \(x_i=x_j\): a normal edge moves coordinate \(i\) or
\(j\), and a tangential edge moves one of the remaining coordinates.

\begin{lemma}
\label{lem:single-diagonal-trace}
There is a universal \(C<\infty\) such that, for every periodic square with
integer \(L\geq3\), every integer \(k\geq2\), every
\(1\leq i<j\leq k\), and every symmetric
\(F\in\ell^2_{\Sym}(B_L^k)\),
\begin{equation}
 \sum_{e\in\mathcal B_{ij}}|\nabla_eF|^2
 \leq C\|L_{0,L,k}F\|^2.
 \label{eq:single-diagonal-trace}
\end{equation}
\end{lemma}

\begin{proof}
Set \(\mathbb Z_L:=\mathbb Z/L\mathbb Z\).  We use the unitary Fourier
transform
\[
 \widehat F(p_1,\ldots,p_k)
 :=V^{-k/2}\sum_{x_1,\ldots,x_k\in B_L}
 e^{-\mathrm i\sum_jp_j\cdot x_j}F(x_1,\ldots,x_k).
\]
Let
\[
 \mathbb T_L^2:=\frac{2\pi}{L}\mathbb Z_L^2,
 \qquad
 \varepsilon(p)=2\sum_{\nu=1}^2(1-\cos p_\nu).
\]
Fix the total momentum \(p\) of particles \(i,j\), their relative momentum
\(q\), and the momenta \(p_r\), \(r\ne i,j\), of the remaining particles,
which we call spectators.  Choose for the
class \(p\) its centred lift \(\widetilde p\in[-\pi,\pi)^2\); all
differences below are understood modulo \(2\pi\) with this lift.  On the
subspace with these momenta fixed, the product eigenvalue is
\begin{equation}
 \lambda(q):=\varepsilon(q)+\varepsilon(\widetilde p-q)+a,
 \qquad a:=\sum_{r\ne i,j}\varepsilon(p_r).
 \label{eq:product-fibre-energy}
\end{equation}
We abbreviate the resulting function of the remaining relative momentum
by \(\widehat F(q)\); the fixed values \(p,(p_r)_{r\ne i,j}\) are
suppressed.
The elementary identity
\[
 \varepsilon(q)+\varepsilon(\widetilde p-q)
 =4\sum_{\nu=1}^2
 \left[1-\cos(\widetilde p_\nu/2)
 \cos(q_\nu-\widetilde p_\nu/2)\right]
\]
implies, with constants independent of \(L,p,q\), that
\begin{equation}
 c\bigl\{\varepsilon(\widetilde p)
       +\varepsilon(q-\widetilde p/2)\bigr\}
 \leq \varepsilon(q)+\varepsilon(\widetilde p-q)
 \leq C\bigl\{\varepsilon(\widetilde p)
       +\varepsilon(q-\widetilde p/2)\bigr\}.
 \label{eq:pair-dispersion-comparison}
\end{equation}
Indeed, for one coordinate put
\(u=\cos(\widetilde p_\nu/2)\in[0,1]\) and
\(r=q_\nu-\widetilde p_\nu/2\).  Then
\[
 4(1-u\cos r)=4(1-u)+4u(1-\cos r),
 \qquad
 2(1-\cos\widetilde p_\nu)=4(1-u)(1+u).
\]
If \(u\geq1/2\), the second term in the first identity controls
\(1-\cos r\); if \(u<1/2\), the first term is bounded below by \(2\)
and hence controls the bounded quantity \(1-\cos r\).  The upper bound
follows from the same identities.

We shall use the following shifted lattice-sum estimate.  For every fixed
\(c_0>0\), uniformly in \(\theta\in\mathbb R^2\), \(L\geq3\), and
\(h\geq0\) satisfying either \(h=0\) or \(h\geq c_0L^{-2}\),
\begin{equation}
 \frac1V\sum_{q\in\mathbb T_L^2:\,h+\varepsilon(q-\theta)>0}
 \frac{h}{\bigl(h+\varepsilon(q-\theta)\bigr)^2}
 \leq C_{c_0}.
 \label{eq:shifted-lattice-resolvent}
\end{equation}
Here and below \(\varepsilon\) is evaluated modulo \(2\pi\).  To prove
\eqref{eq:shifted-lattice-resolvent}, let
\(d_{\mathbb T}(q,\theta)\) be Euclidean distance on the two-dimensional
torus.  On a centred fundamental domain,
\[
 c d_{\mathbb T}(q,\theta)^2
 \leq\varepsilon(q-\theta)
 \leq C d_{\mathbb T}(q,\theta)^2.
\]
Since \(\mathbb T_L^2\) has mesh \(2\pi/L\), a disk-packing argument gives,
uniformly in the shift,
\begin{equation}
 \#\{q\in\mathbb T_L^2:\varepsilon(q-\theta)\leq s\}
 \leq C(1+Vs),\qquad 0<s\leq8.
 \label{eq:shifted-grid-count}
\end{equation}
For completeness, place around every lattice point a disjoint square of
side \(\pi/L\).  All squares corresponding to momenta counted on the
left-hand side of \eqref{eq:shifted-grid-count} lie in a disk of radius
\(C\sqrt{s}+C/L\); comparison of areas gives
\eqref{eq:shifted-grid-count}.

If \(h=0\), every summand in \eqref{eq:shifted-lattice-resolvent} is zero.
If \(h\geq1\), every summand is at most \(h^{-1}\leq1\).  It remains to
consider \(c_0L^{-2}\leq h<1\).  Split the grid into
\[
 A_0=\{\varepsilon(q-\theta)\leq h\},\qquad
 A_j=\{2^{j-1}h<\varepsilon(q-\theta)\leq2^jh\},\quad j\geq1,
\]
stopping once \(2^jh\geq8\).  By \eqref{eq:shifted-grid-count}, the
contribution of \(A_0\) is at most
\[
 \frac{C(1+Vh)}{Vh}\leq C_{c_0},
\]
and the contribution of \(A_j\) is at most
\[
 \frac{C(1+V2^jh)}{V}\,
 \frac{h}{(2^{j-1}h)^2}
 \leq C_{c_0}4^{-j}+C2^{-j}.
\]
Summing in \(j\) proves \eqref{eq:shifted-lattice-resolvent}.  This proof
does not require \(\theta\) to be a lattice momentum; in particular it
covers the half-grid shift \(\theta=\widetilde p/2\).

Put
\[
 \mu:=a+\varepsilon(\widetilde p),\qquad
 \rho(q):=\varepsilon(q-\widetilde p/2).
\]
By \eqref{eq:pair-dispersion-comparison},
\(\lambda(q)\geq c(\mu+\rho(q))\).  Moreover, since \(p\) and all
spectator momenta lie in \(\mathbb T_L^2\), either \(\mu=0\) or
\(\mu\geq c_0L^{-2}\).  The smallest positive value of
\(\varepsilon\) on \(\mathbb T_L^2\) is
\(4\sin^2(\pi/L)\geq c_0L^{-2}\).

Restriction to \(x_i=x_j\) is the normalized sum
\(V^{-1/2}\sum_q\).  For any multiplier \(m(q)\), weighted
Cauchy--Schwarz gives
\begin{equation}
 \left|\frac1{\sqrt V}\sum_qm(q)\widehat F(q)\right|^2
 \leq
 \left(\frac1V\sum_{\lambda(q)>0}
 \frac{|m(q)|^2}{\lambda(q)^2}\right)
 \sum_q\lambda(q)^2|\widehat F(q)|^2.
 \label{eq:weighted-restriction-cs}
\end{equation}
Terms with \(\lambda(q)=0\) are omitted; the corresponding gradient
multiplier vanishes in every application below.

For tangential edges, Parseval in the spectator coordinate gives the sum
of squared multipliers
\[
 \sum_{r\ne i,j}\sum_{|e|=1}
 |e^{\mathrm i p_r\cdot e}-1|^2=2a,
\]
where using oriented rather than unoriented unit vectors changes only the
displayed factor.  Weighted Cauchy--Schwarz therefore reduces
the tangential contribution to
\begin{align}
 \frac1V\sum_{q:\lambda(q)>0}\frac{a}{\lambda(q)^2}
 &\leq \frac C V\sum_q\frac{a}{(\mu+\rho(q))^2} \notag\\
 &\leq \frac C V\sum_q\frac{\mu}{(\mu+\rho(q))^2}
 \leq C.
 \label{eq:tangential-restriction-sum}
\end{align}
If \(\mu=0\), then \(a=0\), so the left-hand side of
\eqref{eq:tangential-restriction-sum} vanishes and no \(0/0\) term is
present.

For a normal edge in a coordinate direction, permutation symmetry gives
\(\widehat F(q,\widetilde p-q)=
  \widehat F(\widetilde p-q,q)\).  Thus, after pairing the terms
\(q\) and \(\widetilde p-q\), the raw multiplier may be replaced exactly
by
\begin{equation}
 m_N(\widetilde p,q)
 :=\frac12\bigl(e^{\mathrm iq\cdot e}
 +e^{\mathrm i(\widetilde p-q)\cdot e}-2\bigr).
 \label{eq:normal-multiplier}
\end{equation}
With \(r=q-\widetilde p/2\), equation \eqref{eq:normal-multiplier} becomes
\(m_N=e^{\mathrm i\widetilde p_\nu/2}\cos r_\nu-1\).  Hence
\[
 |m_N|^2
 \leq C\bigl(\varepsilon(\widetilde p)+\rho(q)^2\bigr).
\]
Consequently,
\begin{align}
 &\frac1V\sum_{q:\lambda(q)>0}
 \frac{\varepsilon(\widetilde p)+\rho(q)^2}{\lambda(q)^2}
 \notag\\
 &\quad\leq
 \frac C V\sum_q\frac{\mu}{(\mu+\rho(q))^2}
 +\frac C V\sum_q\frac{\rho(q)^2}{(\mu+\rho(q))^2}
 \leq C.
 \label{eq:normal-restriction-sum}
\end{align}
The first sum is bounded by \eqref{eq:shifted-lattice-resolvent}; every
summand in the second is at most one.  When \(\mu=0\), the unique possible
zero of \(\lambda\) is omitted and its normal multiplier vanishes.
Summing the finitely many normal and tangential directions and using
Parseval proves \eqref{eq:single-diagonal-trace}.
\end{proof}

Every bad product edge belongs to at least one \(\mathcal B_{ij}\): choose
a colliding pair at a collision endpoint.  Moving a member of that pair is
normal, while moving another particle is tangential.  Triple and higher
collisions only lead to overcounting.  Summing
Lemma~\ref{lem:single-diagonal-trace} over the at most \(k^2/2\) pairs gives
\begin{equation}
 q_{\mathrm{bad}}(F)
 \leq Ck^2\|L_{0,L,k}F\|^2.
 \label{eq:all-collision-strata}
\end{equation}

Multiplying the block matrix in
\eqref{eq:block-product-laplacian} by the harmonic extension gives
\begin{align}
 L_{0,L,k}E_{L,k}f
 &=
 \begin{pmatrix}A&B\\B^*&C_0\end{pmatrix}
 \binom{f}{-C_0^{-1}B^*f} \notag\\
 &=\binom{(A-BC_0^{-1}B^*)f}{0}
 =\binom{(H_{L,k}+N_{L,k})f}{0}.
 \label{eq:L-on-harmonic-extension}
\end{align}
In particular,
\begin{equation}
 \|L_{0,L,k}E_{L,k}f\|^2
 =\|(H_{L,k}+N_{L,k})f\|^2.
 \label{eq:L-extension-norm}
\end{equation}

The quadratic form of \(L_{0,L,k}\) splits into edges having both
endpoints in \(D_{L,k}\) and edges having at least one collision endpoint.
Since \(E_{L,k}f=f\) on \(D_{L,k}\), the first contribution is exactly the
exclusion form.  Hence
\begin{equation}
 \langle E_{L,k}f,L_{0,L,k}E_{L,k}f\rangle
 =\langle f,H_{L,k}f\rangle+q_{\mathrm{bad}}(E_{L,k}f).
 \label{eq:extension-good-bad-splitting}
\end{equation}
On the other hand, \eqref{eq:extension-energy} makes the left-hand side
equal to \(\langle f,(H_{L,k}+N_{L,k})f\rangle\).  Subtracting
\(\langle f,H_{L,k}f\rangle\) proves
\begin{equation}
 q_{\mathrm{bad}}(E_{L,k}f)=\langle f,N_{L,k}f\rangle.
 \label{eq:N-as-bad-energy}
\end{equation}
Applying \eqref{eq:all-collision-strata} to \(E_{L,k}f\) and using
\eqref{eq:L-extension-norm} yields
\[
 \langle f,N_{L,k}f\rangle
 \leq Ck^2\langle f,(H_{L,k}+N_{L,k})^2f\rangle.
\]
Since this holds for every \(f\), we obtain
\begin{equation}
 N_{L,k}\leq Ck^2(H_{L,k}+N_{L,k})^2.
 \label{eq:quadratic-N-prebound}
\end{equation}

\subsection{Bounding the collision multiplier by effective resistance}
\label{subsec:contact-resistance}

Equation~\eqref{eq:quadratic-N-prebound} still contains \(N_{L,k}\) on its
right-hand side.  To remove it we first prove a preliminary first-order bound.
Let \(\nu_{L,k}\) be the uniform probability measure on
\(\Omega_{L,k}\), and let \(R_{L,k}(A)\) count oriented nearest-neighbour
pairs \(a\sim b\) with \(a,b\in A\).  The following estimate is the only
place where the logarithmic two-dimensional resistance enters.
We use the normalized exclusion form
\begin{equation}
 \cE_{L,k}(f):=\langle f,H_{L,k}f\rangle_{\nu_{L,k}}.
 \label{eq:normalized-exclusion-form}
\end{equation}
For a unit-conductance graph \(G\),
\(\mathsf R_{\mathrm{eff}}^G(x,y)\) denotes the electrical effective
resistance between \(x\) and \(y\).  By Thomson's principle it is the
minimum of \(\sum_e|j(e)|^2\) over antisymmetric edge flows \(j\) whose
divergence equals \(1\) at \(x\), \(-1\) at \(y\), and \(0\) elsewhere.
We write \(\dist\) for graph distance.

\begin{lemma}
\label{lem:punctured-resistance}
There is a universal constant \(C<\infty\) such that, for every integer
\(L\geq3\), for either the periodic or the free square, and for every
\(a\in B_L\),
\begin{equation}
 \max_{x,y\in B_L\setminus\{a\}}
 \mathsf R_{\mathrm{eff}}^{B_L\setminus\{a\}}(x,y)
 \leq C\log(2L).
 \label{eq:punctured-resistance}
\end{equation}
\end{lemma}

\begin{proof}
We first prove the logarithmic resistance bound on the undeleted free
square \(P_L\mathbin\square P_L\), where \(P_L\) is the path on
\(\{1,\ldots,L\}\).  An orthonormal eigenbasis of the path Laplacian is
\[
 \phi_m(r)=\left(\frac{2-\delta_{m0}}L\right)^{1/2}
 \cos\!\left(\frac{\pi m(r-1/2)}L\right),
 \quad
 \lambda_m=2\left(1-\cos\frac{\pi m}L\right),
 \quad 0\leq m<L.
\]
Thus \(\psi_{mn}(r,s)=\phi_m(r)\phi_n(s)\) diagonalizes the square
Laplacian, with eigenvalue \(\lambda_m+\lambda_n\).  The spectral formula
for resistance gives, for vertices \(x,y\),
\[
 \mathsf R_{\mathrm{eff}}^{B_L}(x,y)
 =\sum_{(m,n)\ne(0,0)}
 \frac{|\psi_{mn}(x)-\psi_{mn}(y)|^2}{\lambda_m+\lambda_n}.
\]
Since \(|\psi_{mn}|\leq2/L\) and
\(\lambda_m\geq c m^2/L^2\),
\[
 \mathsf R_{\mathrm{eff}}^{B_L}(x,y)
 \leq C\sum_{\substack{0\leq m,n<L\\(m,n)\ne(0,0)}}
 \frac1{m^2+n^2}
 \leq C\log(2L).
\]
The final inequality follows by grouping the integer pairs into dyadic
annuli; each annulus contributes at most a universal constant.

We next show that deletion of one vertex changes this bound by at most a
universal factor.  If \(x=y\), there is nothing to prove.  Fix distinct
\(x,y\ne a\), and let \(j\) be the unit electrical flow from \(x\) to
\(y\) in the undeleted free square.  Orient every edge both ways, write
\(j(u,v)=-j(v,u)\), and use
\(\operatorname{div}j(u)=\sum_{v:v\sim u}j(u,v)\).  Because \(a\) is
neither source nor sink,
\[
 c_z:=j(a,z),\qquad z\sim a,qquad \sum_{z\sim a}c_z=0.
\]
Let \(P_+=\{z:c_z>0\}\) and \(P_-=\{w:c_w<0\}\).  Choose numbers
\(\alpha_{wz}\geq0\), \(w\in P_-\), \(z\in P_+\), such that
\[
 \sum_{z\in P_+}\alpha_{wz}=-c_w,
 \qquad
 \sum_{w\in P_-}\alpha_{wz}=c_z.
\]
Such a transport matrix exists because the two total masses agree.

Any two neighbours of a vertex of the free square can be joined in
\(B_L\setminus\{a\}\) by a path of length at most four.  For an interior
vertex use the surrounding \(3\times3\) square; on a side or at a corner
use its intersection with \(B_L\).  The assumption \(L\geq3\) guarantees
these paths.  For every \((w,z)\in P_-\times P_+\), choose one such path
\(\gamma_{wz}\) and send current \(\alpha_{wz}\) from \(w\) to \(z\)
along it.  Denote the sum of these path flows by \(h\).  Then
\(\operatorname{div}h(z)=-c_z\) at every neighbour \(z\) of \(a\), and
its divergence vanishes elsewhere.  Since there are at most sixteen
paths and each has length at most four,
\[
 \sum_e|h(e)|^2
 \leq C\sum_{w,z}\alpha_{wz}^2
 \leq C\left(\sum_{z\in P_+}c_z\right)^2
 \leq C\sum_{z\sim a}|c_z|^2.
\]
Remove the edges incident to \(a\) from \(j\) and add \(h\) on the
remaining graph.  The resulting flow \(j'\) has divergence \(1\) at
\(x\), \(-1\) at \(y\), and zero elsewhere, and avoids \(a\).  Moreover,
\[
 \sum_e|j'(e)|^2
 \leq2\sum_e|j(e)|^2+2\sum_e|h(e)|^2
 \leq C\sum_e|j(e)|^2
 \leq C\log(2L).
\]
Thomson's principle proves \eqref{eq:punctured-resistance} for the free
square.

Finally, on the same vertex set the periodic square contains every edge
of the free square.  After deleting \(a\), the same inclusion remains
true.  Rayleigh monotonicity, or the use of the free-square flow \(j'\) as
a trial flow in the periodic graph, therefore gives
\[
 \mathsf R_{\mathrm{eff}}^{B_L^{\mathrm{per}}\setminus\{a\}}(x,y)
 \leq
 \mathsf R_{\mathrm{eff}}^{B_L^{\mathrm{free}}\setminus\{a\}}(x,y)
 \leq C\log(2L).
\]
This completes the proof for both boundary conditions.
\end{proof}

\begin{lemma}
\label{lem:contact-bound}
There is a universal constant \(C<\infty\) such that, for every integer
\(L\geq3\), for either choice of boundary conditions
(periodic or free), for every integer
\(1\leq k\leq\lfloor V/2\rfloor\), and for every
\(f\in\ell^2(\Omega_{L,k})\),
\begin{equation}
 \langle f,R_{L,k}f\rangle_{\nu_{L,k}}
 \leq C\frac{k^2}{V}\|f\|_{\nu_{L,k}}^2
 +Ck\log(2L)\langle f,H_{L,k}f\rangle_{\nu_{L,k}}.
 \label{eq:contact-bound}
\end{equation}
\end{lemma}

\begin{proof}
For \(k=1\), one has \(R_{L,1}=0\), so the assertion holds.  Hence assume
\(k\geq2\).  Write \(N_k=\binom Vk\), \(n=V-k+1\), and, for \(b\in A\),
put \(B=A\setminus\{b\}\) and
\[
 \bar f_B:=\frac1n\sum_{y\notin B}f(B\cup\{y\}).
\]
Then
\begin{equation}
 |f(A)|^2\leq2|\bar f_B|^2
 +\frac2n\sum_{y\notin B}|f(A)-f(B\cup\{y\})|^2.
 \label{eq:contact-jensen}
\end{equation}
Let \(d_A(b)=\#\{a\in A:a\sim b\}\).  Since
\(R_{L,k}(A)=\sum_{b\in A}d_A(b)\), averaging
\eqref{eq:contact-jensen} gives
\[
 \langle f,R_{L,k}f\rangle_{\nu_{L,k}}\leq2M+2F,
\]
where
\begin{align*}
 M&:=\frac1{N_k}\sum_A\sum_{b\in A}
 d_A(b)|\bar f_{A\setminus\{b\}}|^2,\\
 F&:=\frac1{N_kn}\sum_A\sum_{b\in A}d_A(b)
 \sum_{y\notin A\setminus\{b\}}
 |f(A)-f((A\setminus\{b\})\cup\{y\})|^2.
\end{align*}

For the mean term, Jensen's inequality and the change of variables
\(A=B\cup\{b\}\) give
\begin{align*}
 M
 &\leq\frac1{N_kn}
 \sum_{\substack{B\subset B_L\\|B|=k-1}}
 \left(\sum_{b\notin B}\#\{a\in B:a\sim b\}\right)
 \sum_{y\notin B}|f(B\cup\{y\})|^2 \\
 &\leq\frac{4(k-1)}{N_kn}
 \sum_{|B|=k-1}\sum_{y\notin B}|f(B\cup\{y\})|^2.
\end{align*}
Every \(k\)-set \(A'\) occurs exactly \(k\) times in the last double
sum, namely as \((B,y)=(A'\setminus\{y\},y)\) with \(y\in A'\).
Consequently,
\begin{equation}
 M\leq\frac{4k(k-1)}n\|f\|_{\nu_{L,k}}^2
 \leq C\frac{k^2}V\|f\|_{\nu_{L,k}}^2,
 \label{eq:contact-mean-bound}
\end{equation}
where \(n\geq V/2\).

We turn to \(F\).  Use the occupied neighbour \(a\) of \(b\) as an
anchor, let \(G_a=B_L\setminus\{a\}\), and condition on \(a\in A\).
Then \(\zeta=A\setminus\{a\}\) is uniform among the \((k-1)\)-subsets
of \(G_a\).  Denote this law by \(\nu^a\), set
\(f_a(\zeta)=f(\zeta\cup\{a\})\), and let
\[
 \cE_a(f_a):=
 \langle f_a,H_{G_a,k-1}f_a\rangle_{\nu^a}.
\]
Since \(\nu_{L,k}(a\in A)=k/V\), the fluctuation term is exactly
\begin{align}
 F={}&\frac{k}{Vn}\sum_{a\in B_L}
 \sum_{\substack{b\in G_a\\b\sim a}}
 \sum_{y\in G_a}
 \mathbb E_{\nu^a}\left[
 \1_{\{b\in\zeta,\,y\notin\zeta\}}
 \left|f_a(\zeta\setminus\{b\}\cup\{y\})-f_a(\zeta)\right|^2
 \right].
 \label{eq:contact-conditioned-exact}
\end{align}
The term \(y=b\) vanishes, so including it does not change
\eqref{eq:contact-conditioned-exact}.

We record the normalization of the moving-particle lemma used here.  If
\(G\) is a finite connected unit-conductance graph, \(\nu_m\) is uniform
on its \(m\)-subsets, and
\(\cE_{G,m}(g)=\langle g,H_{G,m}g\rangle_{\nu_m}\), then Chen's theorem
\cite{Chen2017} gives
\begin{equation}
 \mathbb E_{\nu_m}\left[
 \1_{\{u\in\eta,\,v\notin\eta\}}
 |g(\eta\setminus\{u\}\cup\{v\})-g(\eta)|^2\right]
 \leq \mathsf R_{\mathrm{eff}}^G(u,v)\,\cE_{G,m}(g).
 \label{eq:moving-particle-use}
\end{equation}
The left-hand side of \eqref{eq:moving-particle-use} equals one half of the
expectation of the squared full \(u,v\)-swap, while \(\cE_{G,m}\) is one
half of the sum of the
nearest-neighbour exchange squares.  Thus there is no hidden factor
depending on \(|G|\) or \(m\).

Applying \eqref{eq:moving-particle-use} on \(G_a\), followed by
Lemma~\ref{lem:punctured-resistance}, bounds every expectation in
\eqref{eq:contact-conditioned-exact} by
\(C\log(2L)\cE_a(f_a)\).  There are at most four choices of \(b\), and
\(|G_a|=V-1\); since \((V-1)/n\leq2\),
\begin{equation}
 F\leq C\frac{k\log(2L)}V\sum_{a\in B_L}\cE_a(f_a).
 \label{eq:contact-fluctuation-before-anchor}
\end{equation}

It remains to verify the anchored-energy identity.  In unnormalized form,
every allowed exclusion move \(A\to A\setminus\{u\}\cup\{v\}\) is
present in \(\cE_a(f_a)\) precisely when
\(a\in A\setminus\{u\}\), hence for exactly \(k-1\) anchors.  Therefore
\[
 \binom{V-1}{k-1}\sum_{a\in B_L}\cE_a(f_a)
 =(k-1)\binom Vk\cE_{L,k}(f).
\]
Using \(\binom Vk/\binom{V-1}{k-1}=V/k\) gives
\begin{equation}
 \frac{k}{V}\sum_{a\in B_L}\cE_a(f_a)
 =(k-1)\cE_{L,k}(f).
 \label{eq:anchored-energy-identity}
\end{equation}
Equations \eqref{eq:contact-fluctuation-before-anchor} and
\eqref{eq:anchored-energy-identity} give
\(F\leq Ck\log(2L)\cE_{L,k}(f)\).  Together with
\eqref{eq:contact-mean-bound}, this proves \eqref{eq:contact-bound}.
\end{proof}

By the exclusion-process consequence of Aldous' spectral-gap theorem
\cite{CaputoLiggettRichthammer2010}, the
exclusion gap equals the one-particle random-walk gap.  For a nonnegative
operator whose kernel
consists of constants, \(\gap(T)\) denotes its smallest positive
eigenvalue.  Thus, for periodic or free squares,
\begin{equation}
 \gap(H_{L,k})\geq cL^{-2}=cV^{-1}.
 \label{eq:exclusion-gap}
\end{equation}
Since \(N_{L,k}\leq R_{L,k}\), Lemma~\ref{lem:contact-bound} and
\eqref{eq:exclusion-gap} imply, on the orthogonal complement of the
constant function,
\begin{equation}
 N_{L,k}\leq b_{L,k}H_{L,k},
 \qquad b_{L,k}:=Ck^2\log(2L).
 \label{eq:first-order-N-bound}
\end{equation}

\begin{proof}[Proof of Theorem~\ref{thm:normal-trace}]
The quadratic estimate records the second-order cancellation, while
Lemma~\ref{lem:contact-bound} gives a first-order relative bound.  We now
combine these two operator inequalities.  Work on the orthogonal complement
of constants and put
\begin{equation}
 \mathcal K_{L,k}:=H_{L,k}^{-1/2}N_{L,k}H_{L,k}^{-1/2}.
 \label{eq:relative-collision-operator}
\end{equation}
Then \(0\leq\mathcal K_{L,k}\leq b_{L,k}I\).  Multiplying
\eqref{eq:quadratic-N-prebound} on the left and right by
\(H_{L,k}^{-1/2}\) yields
\[
 \mathcal K_{L,k}
 \leq Ck^2(I+\mathcal K_{L,k})H_{L,k}(I+\mathcal K_{L,k}).
\]
Since \(\mathcal K_{L,k}\) commutes with \(I+\mathcal K_{L,k}\),
multiplying on the left and right by
\((I+\mathcal K_{L,k})^{-1}\) gives
\begin{equation}
 \mathcal K_{L,k}(I+\mathcal K_{L,k})^{-2}\leq Ck^2H_{L,k}.
 \label{eq:absorbed-functional-calculus}
\end{equation}
For \(0\leq x\leq b_{L,k}\),
\(x\leq(1+b_{L,k})^2x(1+x)^{-2}\).  Apply this scalar inequality to the
spectral decomposition of \(\mathcal K_{L,k}\) in
\eqref{eq:absorbed-functional-calculus}.
Multiplying the resulting operator inequality on the left and right by
\(H_{L,k}^{1/2}\) proves
\[
 N_{L,k}\leq Ck^2(1+b_{L,k})^2H_{L,k}^2
 \leq Ck^6\log^2(2L)H_{L,k}^2.
\]
Since \(H_{L,k}\) and \(N_{L,k}\) vanish on constants, the bound
\(N_{L,k}\leq Ck^6\log^2(2L)H_{L,k}^2\) holds on the full exclusion
space.  This proves Theorem~\ref{thm:normal-trace} for periodic boxes.

For free boundary conditions, the proof of Lemma~\ref{lem:contact-bound}
uses only the resistance bound on the free square and therefore yields the
same inequality with the same universal-constant convention.
We transfer the second-order diagonal estimate from a periodic square of
side \(2L\).  Define
\[
 \pi_L(z)=
 \begin{cases}
  z+1,&0\leq z<L,\\
  2L-z,&L\leq z<2L,
 \end{cases}
 \qquad z\in\mathbb Z/(2L\mathbb Z),
\]
and let \(\Pi_L(z_1,z_2)=(\pi_L(z_1),\pi_L(z_2))\).  For
\(u:B_L\to\mathbb C\), set \((Ju)(z)=u(\Pi_Lz)\).  Direct inspection at
the two reflection seams, as well as in the interior, gives
\begin{equation}
 \|Ju\|^2=4\|u\|^2,
 \qquad
 (-\Delta_{\mathrm{per},2L})Ju=J(-\Delta_N)u.
 \label{eq:neumann-reflection-intertwining}
\end{equation}
Tensorizing over the \(k\) particle coordinates, write
\(\widetilde F=J^{\otimes k}F\), and set
\[
 L_{0,L,k}^{N}:=\sum_{i=1}^k(-\Delta_N)_i,
 \qquad
 L_{0,2L,k}^{\mathrm{per}}
 :=\sum_{i=1}^k(-\Delta_{\mathrm{per},2L})_i.
\]
Then
\begin{equation}
 \|L_{0,2L,k}^{\mathrm{per}}\widetilde F\|^2
 =4^k\|L_{0,L,k}^{N}F\|^2.
 \label{eq:reflected-product-L2}
\end{equation}

It remains to identify the lifted collision diagonals.  On
\(\mathbb Z/(2L\mathbb Z)\), let \(r_+(z)=z\) and
\(r_-(z)=-z-1\).  Then \(\pi_L(r_\pm z)=\pi_L(z)\), and
\(\pi_L(z)=\pi_L(w)\) if and only if \(z=w\) or \(z=-w-1\).
For \(\sigma=(\sigma_1,\sigma_2)\in\{+,-\}^2\), put
\[
 r_\sigma(z_1,z_2)
 =(r_{\sigma_1}(z_1),r_{\sigma_2}(z_2)).
\]
The inverse image of the free diagonal \(x_i=x_j\) is the disjoint union
of the four affine periodic diagonals
\[
 D_{ij}^\sigma=\{z_i=r_\sigma z_j\},
 \qquad \sigma\in\{+,-\}^2.
\]
Let \(\mathcal B_{ij}^\sigma\) be the periodic product-edge family with at
least one endpoint in \(D_{ij}^\sigma\), defined in the same way as
\(\mathcal B_{ij}\) for the standard pair diagonal.
The extension \(\widetilde F\) is permutation symmetric and is invariant
under applying any \(r_\sigma\) to any one particle coordinate.  Since
every \(r_\sigma\) is a graph automorphism of the \(2L\)-periodic square,
the change of variable \(z_j\mapsto r_\sigma z_j\) maps the edge family
based on \(D_{ij}^\sigma\) to the standard family based on \(z_i=z_j\),
preserves the product Laplacian, and leaves \(\widetilde F\) unchanged.
Lemma~\ref{lem:single-diagonal-trace} therefore gives, for every
\(\sigma\),
\begin{equation}
 \sum_{e\in\mathcal B_{ij}^\sigma}
 |\nabla_e\widetilde F|^2
 \leq C\|L_{0,2L,k}^{\mathrm{per}}\widetilde F\|^2.
 \label{eq:affine-reflected-diagonal-bound}
\end{equation}

Every nonconstant free product edge has exactly \(4^k\) periodic lifts:
there are four lifts of the moved spatial edge and four choices for every
unmoved particle position.  Edges at a reflection seam which project to a
single vertex have zero \(\widetilde F\)-gradient.  Hence, if
\(\mathcal B_{ij}^{N}\) denotes the free edge family based on
\(x_i=x_j\),
\begin{align}
 4^k\sum_{e\in\mathcal B_{ij}^{N}}|\nabla_eF|^2
 &\leq \sum_{\sigma\in\{+,-\}^2}
       \sum_{e\in\mathcal B_{ij}^\sigma}
       |\nabla_e\widetilde F|^2 \\
 &\leq C\|L_{0,2L,k}^{\mathrm{per}}\widetilde F\|^2
 =C4^k\|L_{0,L,k}^{N}F\|^2.
 \label{eq:free-reflected-diagonal-bound}
\end{align}
Dividing \eqref{eq:free-reflected-diagonal-bound} by \(4^k\) proves
\eqref{eq:single-diagonal-trace} for the free square.  Summing over
the particle pairs and repeating
\eqref{eq:all-collision-strata}--\eqref{eq:absorbed-functional-calculus}
proves Theorem~\ref{thm:normal-trace} for free boundary conditions.
\end{proof}

\section{From collision control to the sharp pressure}
\label{sec:pressure}

We now turn the quadratic collision estimate into the missing free-energy
bound for spin \(1/2\).  The main object is the finite-volume partition
function \(Z_L(t)\) defined below.  We introduce explicit cutoffs on the
energy and the number of reversed spins, show that the discarded trace is
negligible, compare the retained eigenvalues with free Neumann bosons, and
then pass to the thermodynamic limit.
The finite-box reduction and the energy and particle-number cutoffs follow
the general strategy of
\cite{CorreggiGiulianiSeiringer2015,NapiorkowskiSeiringer2021}; the
harmonic-extension comparison is the new ingredient.

\subsection{Low energy and few reversed spins}
\label{subsec:low-energy-cutoff}

For the remainder of this section the boxes have free boundary.  Put
\[
 \mathcal H_L:=2H_{B_L}^{(1/2)},\qquad
 Z_L(t):=\Tr e^{-t\mathcal H_L}.
\]
Thus \(Z_L(\beta/2)\) is the physical spin-\(1/2\) partition function.
We first need a preliminary estimate which is uniform in the box.

\begin{lemma}
\label{lem:rough-sector-energy}
There is a universal constant \(c_{\mathrm{sec}}>0\) such that, for every
integer \(m\geq2\), with \(v=m^2\), and every integer
\(1\leq k\leq\lfloor v/2\rfloor\), the free square satisfies
\begin{equation}
 H_{m,k}^{\HW}\geq c_{\mathrm{sec}}\,\frac{k}{v}.
 \label{eq:rough-sector-energy}
\end{equation}
Consequently, for \(s>0\),
\begin{equation}
 Z_m(s)\leq(v+1)\sum_{k=0}^{\lfloor v/2\rfloor}
 \binom vk e^{-c_{\mathrm{sec}}sk/v}
 \leq(v+1)\exp\{ve^{-c_{\mathrm{sec}}s/v}\}.
 \label{eq:rough-box-trace}
\end{equation}
\end{lemma}

\begin{proof}
Let \(K_{v,k}\) be the exclusion Laplacian of the complete graph on the
same \(v\) sites.  For each pair \(x,y\), fix a Manhattan path
\(\gamma_{xy}=(x_0,\ldots,x_\ell)\), and let \(\tau_e\) exchange the
occupations across an edge \(e\).  If \(\tau_{xy}\) exchanges the
occupations at the endpoints, the path writes
\(\tau_{xy}\) as at most \(2\ell-1\) nearest-neighbour exchanges.  Hence
\[
 \|f-\tau_{xy}f\|^2
 \leq(2\ell-1)\sum_{e\in\gamma_{xy}}
 C\|f-\tau_ef\|^2.
\]
Summing over \(x,y\), the path-length-weighted load of every lattice edge
is \(O(m^4)=O(v^2)\).  Here is the count.  Use the canonical path which
first moves horizontally from \(x=(x_1,x_2)\) to \((y_1,x_2)\), and then
vertically to \(y=(y_1,y_2)\).  A fixed horizontal edge can occur only
when \(x_2\) is its row and \(x_1,y_1\) straddle it; after choosing
\(y_2\), this gives at most \(m^3\) endpoint pairs.  Interchanging the two
coordinate directions gives the identical count for a vertical edge.
Since every path has length at most \(2m\),
the additional factor \(2\ell-1\) makes the weighted load at most
\(Cm^4\).  Therefore
\begin{equation}
 \langle f,K_{v,k}f\rangle
 \leq Cv^2\langle f,H_{m,k}f\rangle.
 \label{eq:complete-nearest-comparison}
\end{equation}
For completeness, we record the complete-graph total-spin identity.  On
\(n\) spin-half sites let \(\tau_{xy}\) exchange the two spin factors and
put
\[
 \mathbb K_n:=\sum_{1\leq x<y\leq n}(I-\tau_{xy}).
\]
Since
\[
 \tau_{xy}=2\mathbf s_x\mathbin\cdot\mathbf s_y+\frac12I,
 \qquad \mathbf S_{\mathrm{tot}}:=\sum_{x=1}^n\mathbf s_x,
\]
one obtains
\begin{equation}
 \mathbb K_n=\frac{n(n+2)}4I-\mathbf S_{\mathrm{tot}}^2.
 \label{eq:complete-graph-casimir}
\end{equation}
On a total-spin-\(J\) irreducible component this is the scalar
\(n(n+2)/4-J(J+1)\).  With \(J=n/2-k\), it equals
\begin{equation}
 \frac{n(n+2)}4-
 \left(\frac n2-k\right)\left(\frac n2-k+1\right)
 =k(n-k+1).
 \label{eq:complete-graph-casimir-sector}
\end{equation}
The restriction of \(\mathbb K_v\) to the \(k\)-down-spin sector is
\(K_{v,k}\).  Since \(k\leq v/2\),
\[
 K_{v,k}=k(v-k+1)I\geq\frac{kv}{2}I
\]
on the highest-weight multiplicity space.  Combining this with
\eqref{eq:complete-nearest-comparison} proves
\eqref{eq:rough-sector-energy}.
The highest-weight multiplicity is at most \(\binom vk\), every spin
multiplet has dimension at most \(v+1\), and summing the binomial series
gives \eqref{eq:rough-box-trace}.
\end{proof}

Let \(t\to\infty\) and choose, up to integer rounding,
\begin{equation}
 g:=t^{1/20},\qquad L^2=:V=tg.
 \label{eq:thermal-main-box}
\end{equation}
Partition the \(L\)-box into squares of volume
\begin{equation}
 m^2=\frac{t}{A\log t},
 \label{eq:rough-cell-size}
\end{equation}
where \(A\) is a sufficiently large fixed constant.  If divisibility
fails, choose nearby integers with \(L\) a multiple of \(m\); the relative
changes are \(1+o(1)\), and none of the estimates below is affected.
Dropping the positive bonds between cells and using
\eqref{eq:rough-box-trace} at time \(t/2\) gives
\begin{equation}
 \log Z_L(t/2)\leq Cg(\log t)^2.
 \label{eq:rough-main-box-trace}
\end{equation}
Indeed, there are \(V/m^2=O(g\log t)\) cells.  For one cell,
\(v=m^2\asymp t/(A\log t)\), and hence
\[
 \frac{c_{\mathrm{sec}}(t/2)}v
 \geq \frac{c_{\mathrm{sec}}A}{3}\log t
\]
for all large \(t\), for the integer choices just described.  Choose \(A\)
so that the
last coefficient is at least \(3\).  Then
\(ve^{-c_{\mathrm{sec}}t/(2v)}\leq Ct^{-2}\), and
\eqref{eq:rough-box-trace} gives \(\log Z_m(t/2)\leq C\log t\).
Multiplication by the number of cells proves
\eqref{eq:rough-main-box-trace}.

Choose
\begin{equation}
 e_0:=C_{\mathrm{cut}}\frac{g(\log t)^2}{t},
 \label{eq:energy-cutoff}
\end{equation}
where \(C_{\mathrm{cut}}\) is larger than twice the constant in
\eqref{eq:rough-main-box-trace}.  For every bounded Borel function
\(F:[0,\infty)\to\mathbb C\), define
\[
 \Tr_{\HW}F(\mathcal H_L)
 :=\sum_{k=0}^{\lfloor V/2\rfloor}
 \Tr_{\mathcal M_{L,k}}F(H_{L,k}^{\HW}).
\]
Thus the full spin trace is bounded by \(V+1\) times
\(\Tr_{\HW}\).  For every scalar \(E\geq0\), the spectral inequality
\(\1_{\{E>e_0\}}e^{-tE}\leq e^{-te_0/2}e^{-tE/2}\) gives
\begin{equation}
 \Tr_{\HW}\1_{\{\mathcal H_L>e_0\}}e^{-t\mathcal H_L}
 \leq e^{-te_0/2}Z_L(t/2)
 \leq e^{-cg(\log t)^2}.
 \label{eq:high-energy-tail}
\end{equation}
Here \(\1_{\{\mathcal H_L>e_0\}}\) is the spectral projection of
\(\mathcal H_L\) onto \((e_0,\infty)\).
The omitted factor \(V+1\) in the corresponding full trace is negligible
because \(\log(V+1)=O(\log t)=o(g(\log t)^2)\).
Applying \eqref{eq:rough-sector-energy} with \(m=L\) shows that an
eigenvector of energy at most \(e_0\) can occur only when
\begin{equation}
 k\leq K(t):=\left\lfloor c_{\mathrm{sec}}^{-1}Ve_0\right\rfloor
 \leq Cg^2(\log t)^2.
 \label{eq:magnon-number-cutoff}
\end{equation}

\subsection{Min--max comparison and the zero mode}
\label{subsec:minmax-zero-mode}

The highest-weight multiplicity space \(\cM_{L,0}\) is one-dimensional and
\(H_{L,0}^{\HW}=0\).  Hence
\(\Tr_{\cM_{L,0}}e^{-tH_{L,0}^{\HW}}=1\).  In the full trace
\eqref{eq:su2-trace-decomposition}, its \(V+1\) magnetic states are supplied
by the prefactor \(V-2k+1\).  We treat this sector separately.  Fix an integer
\(1\leq k\leq K(t)\), and let \(P_{k,e_0}\) be the spectral
projection of
\(H_{L,k}^{\HW}\) onto \([0,e_0]\).  The collision estimate is used only
after both sides are tested on the range of \(P_{k,e_0}\); no invariance
of this range under \(N_{L,k}\) is assumed.
On this range the spectral theorem gives
\(P_{k,e_0}H_{L,k}^2P_{k,e_0}
 \leq e_0P_{k,e_0}H_{L,k}P_{k,e_0}\).
Theorem~\ref{thm:normal-trace} therefore gives
\begin{equation}
 P_{k,e_0}(H_{L,k}+N_{L,k})P_{k,e_0}
 \leq(1+\delta_t)P_{k,e_0}H_{L,k}P_{k,e_0},
 \qquad
 \delta_t:=C_{\mathrm{tr}}K(t)^6\log^2(2L)e_0.
 \label{eq:low-energy-relative-bound}
\end{equation}
The choices \eqref{eq:thermal-main-box}, \eqref{eq:energy-cutoff}, and
\eqref{eq:magnon-number-cutoff} yield
\begin{equation}
 \delta_t\leq Ct^{-7/20}(\log t)^{16}=o(1).
 \label{eq:delta-small}
\end{equation}

Let
\(\lambda_{k,1}\leq\cdots\leq\lambda_{k,M_k}\leq e_0\) be precisely the
eigenvalues of \(H_{L,k}^{\HW}\) not exceeding \(e_0\), counted with
multiplicity, and let \(\mu_{k,j}\) be the increasingly ordered
eigenvalues of the full symmetric free product \(L_{0,L,k}\),
\(1\leq j\leq M_k\).  Apply min--max to the harmonic extensions of the
first \(j\) interacting eigenvectors.  These extensions remain linearly
independent because restriction to \(D_{L,k}\) is a left inverse of
\(E_{L,k}\).  By
\eqref{eq:extension-norm}, \eqref{eq:extension-energy}, and
\eqref{eq:low-energy-relative-bound},
\begin{equation}
 \mu_{k,j}\leq(1+\delta_t)\lambda_{k,j}.
 \label{eq:minmax-eigenvalue-comparison}
\end{equation}
Here \(\Sym^k\ell^2(B_L)\) denotes the \(k\)-fold symmetric tensor power
of the one-particle space.  Thus, with \(s=t/(1+\delta_t)\),
\begin{equation}
 \Tr_{\cM_{L,k}}\1_{\{H_{L,k}^{\HW}\leq e_0\}}e^{-tH_{L,k}^{\HW}}
 \leq\Tr_{\Sym^k\ell^2(B_L)}e^{-sL_{0,L,k}}.
 \label{eq:canonical-trace-comparison}
\end{equation}

Let \(Z_j^{\ne0}(s)\) be the free canonical trace, meaning the trace at
fixed particle number \(j\), with all particles in nonzero one-particle
Neumann modes, and let \(Z_k^0(s)\) include the unique zero mode.
Occupying that zero mode with \(k-j\) particles gives
\[
 Z_k^0(s)=\sum_{j=0}^k Z_j^{\ne0}(s).
\]
Consequently,
\begin{equation}
 \sum_{k=0}^{K(t)} Z_k^0(s)
 \leq(K(t)+1)
 \prod_{\eta\in\spec(-\Delta_N)\setminus\{0\}}
 (1-e^{-s\eta})^{-1},
 \label{eq:zero-mode-grand-canonical}
\end{equation}
where the product sums all occupation numbers independently and therefore
is the grand-canonical partition function over the positive eigenvalues of
the one-particle Neumann Laplacian.  Combining
\eqref{eq:su2-trace-decomposition}, \eqref{eq:high-energy-tail},
\eqref{eq:canonical-trace-comparison}, and
\eqref{eq:zero-mode-grand-canonical}, we find
\begin{equation}
 \log Z_L(t)\leq\log(V+1)+\log(K(t)+1)
 -\sum_{\eta\in\spec(-\Delta_N)\setminus\{0\}}
 \log(1-e^{-s\eta})+o(1).
 \label{eq:box-pressure-before-heat}
\end{equation}

\subsection{The Neumann pressure and the thermodynamic limit}
\label{subsec:neumann-pressure}

The remaining expression is an ideal Bose pressure.  The following
uniform estimate keeps the boundary and lattice errors below the scale of
the main term.

\begin{lemma}
\label{lem:neumann-heat-trace}
For every integer \(L\geq2\), let
\[
 \varepsilon_{m,n}:=
 2\bigl(1-\cos(\pi m/L)\bigr)
 +2\bigl(1-\cos(\pi n/L)\bigr),
 \qquad 0\leq m,n\leq L-1.
\]
Uniformly for \(1\leq s\leq L^2\),
\begin{equation}
 -\sum_{(m,n)\ne(0,0)}\log(1-e^{-s\varepsilon_{m,n}})
 =\frac{\pi L^2}{24s}
 +O\!\left(
 \frac{L}{\sqrt s}
 +\log\!\left(2+\frac{L^2}{s}\right)
 +\frac{L^2}{s^2}+1\right).
 \label{eq:neumann-heat-trace}
\end{equation}
\end{lemma}

\begin{proof}
Write
\[
 \Theta_L(u):=\sum_{m=0}^{L-1}
 e^{-2u(1-\cos(\pi m/L))}.
\]
Gaussian comparison of \(2(1-\cos q)\) with \(q^2\), followed by a
Riemann-sum estimate, gives, for \(1\leq u\leq L^2\),
\begin{equation}
 \Theta_L(u)=\frac{L}{\sqrt{4\pi u}}
 +O\!\left(1+\frac{L}{u^{3/2}}\right).
 \label{eq:one-dimensional-neumann-trace}
\end{equation}
Indeed, for \(|q|\leq\pi\),
\[
 \left|e^{-2u(1-\cos q)}-e^{-uq^2}\right|
 \leq Cuq^4e^{-cuq^2}.
\]
Summing this estimate at \(q=\pi m/L\) contributes
\(O(Lu^{-3/2})\), while comparison of the decreasing Gaussian sum with
its integral contributes \(O(1)\).  For \(u\geq L^2\), the modewise bound
\[
 2\bigl(1-\cos(\pi m/L)\bigr)\geq c\,m^2/L^2
\]
implies
\[
 \Theta_L(u)-1
 \leq\sum_{m\geq1}e^{-cu m^2/L^2}
 \leq Ce^{-c'u/L^2}.
\]

Expanding the logarithm and using monotone convergence gives
\begin{equation}
 -\sum_{(m,n)\ne(0,0)}\log(1-e^{-s\varepsilon_{m,n}})
 =\sum_{\ell\geq1}\frac{\Theta_L(\ell s)^2-1}{\ell}.
 \label{eq:heat-log-expansion}
\end{equation}
Put
\[
 R:=\frac{L^2}{s}\geq1,
 \qquad M:=\lfloor R\rfloor.
\]
For \(1\leq\ell\leq M\), write
\[
 \Theta_L(\ell s)
 =\frac{L}{\sqrt{4\pi\ell s}}+\mathcal R_\ell,
 \qquad
 |\mathcal R_\ell|
 \leq C\left(1+\frac{L}{(\ell s)^{3/2}}\right).
\]
Since \(\ell s\geq1\), expanding the square gives
\begin{align}
 \left|
 \Theta_L(\ell s)^2-1-\frac{L^2}{4\pi\ell s}
 \right|
 &\leq C\left(
 1+\frac{L}{\sqrt{\ell s}}+\frac{L^2}{(\ell s)^2}
 \right).
 \label{eq:heat-trace-square-error}
\end{align}
Indeed,
\[
 2\frac{L}{\sqrt{4\pi\ell s}}|\mathcal R_\ell|
 \leq C\left(\frac{L}{\sqrt{\ell s}}
 +\frac{L^2}{(\ell s)^2}\right),
\]
while
\[
 |\mathcal R_\ell|^2
 \leq C\left(1+\frac{L^2}{(\ell s)^3}\right)
 \leq C\left(1+\frac{L^2}{(\ell s)^2}\right).
\]
Consequently,
\begin{align}
 &\sum_{\ell=1}^{M}\frac1\ell
 \left|
 \Theta_L(\ell s)^2-1-\frac{L^2}{4\pi\ell s}
 \right| \notag\\
 &\quad\leq C\sum_{\ell=1}^{M}\frac1\ell
 +C\frac{L}{\sqrt s}\sum_{\ell=1}^{M}\ell^{-3/2}
 +C\frac{L^2}{s^2}\sum_{\ell=1}^{M}\ell^{-3} \notag\\
 &\quad\leq C\left[
 \log\!\left(2+\frac{L^2}{s}\right)
 +\frac{L}{\sqrt s}+\frac{L^2}{s^2}\right].
 \label{eq:heat-trace-summed-error}
\end{align}
The truncated main term differs from the full Bose sum by a bounded
amount.  Indeed,
\[
 0\leq\sum_{\ell>M}\frac1{\ell^2}\leq\frac{C}{M},
 \qquad \frac{R}{M}\leq2,
\]
and therefore
\begin{equation}
 \frac{L^2}{4\pi s}\sum_{\ell=1}^{M}\frac1{\ell^2}
 =\frac{L^2}{4\pi s}\sum_{\ell\geq1}\frac1{\ell^2}+O(1)
 =\frac{\pi L^2}{24s}+O(1).
 \label{eq:heat-trace-main-term}
\end{equation}

It remains to control \(\ell>M\).  Since \(\ell s>L^2\), the large-\(u\)
estimate above gives
\[
 0\leq\Theta_L(\ell s)^2-1
 \leq Ce^{-c\ell s/L^2}=Ce^{-c\ell/R}.
\]
Moreover, \(M/R\geq1/2\).  Comparison with an integral yields
\begin{equation}
 \sum_{\ell>M}\frac{\Theta_L(\ell s)^2-1}{\ell}
 \leq C\sum_{\ell>M}\frac{e^{-c\ell/R}}{\ell}
 \leq C\int_{M/R}^{\infty}\frac{e^{-cu}}{u}\,du+C
 \leq C.
 \label{eq:heat-trace-large-ell}
\end{equation}
Combining \eqref{eq:heat-log-expansion}--
\eqref{eq:heat-trace-large-ell} proves
\eqref{eq:neumann-heat-trace}.
\end{proof}

In our choice of scales, \(L^2=tg\) and \(s=t(1+o(1))\).  Therefore
\[
 \frac{L}{\sqrt s}=O(\sqrt g),\qquad
 \log\!\left(2+\frac{L^2}{s}\right)=O(\log g),\qquad
 \frac{L^2}{s^2}=O(g/t),
\]
all of which are \(o(g)=o(V/t)\).  Also
\(\log V+\log(K(t)+1)=o(g)\).  Equations
\eqref{eq:box-pressure-before-heat} and
\eqref{eq:neumann-heat-trace} hence give
\begin{equation}
 \log Z_L(t)\leq\frac{\pi}{24}\frac{V}{t}+o(V/t).
 \label{eq:sharp-spinhalf-box-pressure}
\end{equation}

Tile a square of side \(nL\) by \(n^2\) translates of \(B_L\), and write
its Hamiltonian as \(\mathcal H_0+W\), where \(\mathcal H_0\) is the sum
of the box Hamiltonians and \(W\geq0\) is the sum of the interbox bond
operators.  For \(0\leq u\leq1\), Duhamel's formula and cyclicity of the
trace give
\begin{equation}
 \frac{d}{du}\log\Tr e^{-t(\mathcal H_0+uW)}
 =-t\,\frac{\Tr\bigl(W e^{-t(\mathcal H_0+uW)}\bigr)}
              {\Tr e^{-t(\mathcal H_0+uW)}}\leq0.
 \label{eq:bond-dropping-duhamel}
\end{equation}
Thus removing the interbox bonds increases the partition function, and
\(\log Z_{nL}(t)\leq n^2\log Z_L(t)\).  Before taking the two limits, set
\[
 \mathcal H_\Lambda:=2H_\Lambda^{(1/2)},\qquad
 p_{1/2}(t):=\lim_{\Lambda\uparrow\Z^2}
 \frac1{|\Lambda|}\log\Tr e^{-t\mathcal H_\Lambda}.
\]
First taking \(n\to\infty\), and then \(t\to\infty\), in
\eqref{eq:sharp-spinhalf-box-pressure} gives
\[
 \limsup_{t\to\infty}t\,p_{1/2}(t)\leq\frac{\pi}{24}.
\]
Since \(t=\beta/2\), this is
\begin{equation}
 \liminf_{\beta\to\infty}
 \beta^2\frac12 f_2(\beta,\tfrac12)\geq-\frac\pi{24}.
 \label{eq:spinhalf-new-direction}
\end{equation}
This proves the new direction of Theorem~\ref{thm:main-free-energy} for
spin \(1/2\).

\section{Arbitrary fixed spin through a decorated graph}
\label{sec:fixed-spin}

It remains to pass from spin \(1/2\) to arbitrary fixed \(S\).  The main
object is the square lattice decorated by \(r=2S\) spin-half fibres at
each site.  Its one-particle space splits into functions constant in the
fibre variable and functions orthogonal to them.  On the first subspace the
operator is the scalar lattice Laplacian; on the second it has a positive
lower bound independent of the box size.  We prove the decorated analogue
of Theorem~\ref{thm:normal-trace} and evaluate the two spectral contributions
separately.

\subsection{The decorated graph and its one-particle decomposition}
\label{subsec:decorated-graph}

For this subsection \(B_L\) may have periodic or free boundary conditions.
Put \(r:=2S\), write \([r]:=\{1,\ldots,r\}\), and set
\(\mathcal V_{L,r}:=B_L\times[r]\).  We use the weighted decorated graph
\(\Gamma_{L,r}\) with conductances
\begin{equation}
 c_{(x,a),(y,b)}:=\frac1r\,\1_{\{x\sim y\}},
 \qquad (x,a),(y,b)\in\mathcal V_{L,r}.
 \label{eq:decorated-conductances}
\end{equation}
Thus every base edge is replaced by the complete bipartite graph
\(K_{r,r}\), with conductance \(1/r\) on each of its edges.  The total
conductance out of \((x,a)\) is
\(\deg_{B_L}(x)\), independently of \(r\).  In particular,
\begin{equation}
 \Delta_c:=\max_{\xi\in\mathcal V_{L,r}}
 \sum_{\eta\in\mathcal V_{L,r}}c_{\xi\eta}\leq4.
 \label{eq:decorated-weighted-degree}
\end{equation}

On the unweighted \(K_{r,r}\)-decorated square consider spin-\(1/2\)
generators \(\mathbf s_{x,a}\) and the Hamiltonian
\begin{equation}
 \widetilde H_L
 :=\sum_{\langle x,y\rangle}\sum_{a,b=1}^r
 \left(\frac14-\mathbf s_{x,a}\mathbin\cdot\mathbf s_{y,b}\right).
 \label{eq:decorated-hamiltonian}
\end{equation}
Let \(V_S\cong\mathbb C^{2S+1}\) denote the irreducible spin-\(S\)
representation space.
Let \(\mathcal H_{x,\mathrm{sym}}\subset(\mathbb C^2)^{\otimes r}\) be the
fully symmetric fibre.  It is the irreducible spin-\(S=r/2\)
representation.  Since
\(\sum_{a=1}^r\mathbf s_{x,a}=\Sp_x\) on this subspace, one has,
for every base bond,
\begin{equation}
 \sum_{a,b=1}^r
 \left(\frac14-\mathbf s_{x,a}\mathbin\cdot\mathbf s_{y,b}\right)
 =\frac{r^2}{4}-\Sp_x\mathbin\cdot\Sp_y
 =S^2-\Sp_x\mathbin\cdot\Sp_y.
 \label{eq:fibre-bond-identity}
\end{equation}
The tensor product of the symmetric fibres is invariant under
\(\widetilde H_L\).  Define the physical finite-volume partition function
by
\[
 Z_L^{(S)}(\beta):=
 \Tr_{\bigotimes_{x\in B_L}V_S}e^{-\beta H_{B_L}^{(S)}},
\]
where free boundary conditions are understood.  Under the natural unitary
identification \(\mathcal H_{x,\mathrm{sym}}\cong V_S\),
\begin{equation}
 Z_L^{(S)}(\beta)
 =\Tr_{\otimes_x\mathcal H_{x,\mathrm{sym}}}
 e^{-\beta\widetilde H_L}
 \leq \Tr_{(\mathbb C^2)^{\otimes rV}}e^{-\beta\widetilde H_L}.
 \label{eq:physical-decorated-trace}
\end{equation}

For a spin-half ferromagnet, twice the Hamiltonian restricted to a fixed
number of down spins is the exclusion generator with unit conductances.
Consequently
\begin{equation}
 \mathcal H_L^{(r)}:=\frac2r\widetilde H_L
 \label{eq:normalized-decorated-hamiltonian}
\end{equation}
is, in every down-spin sector, the exclusion generator on the weighted
graph \(\Gamma_{L,r}\) of \eqref{eq:decorated-conductances}.  Explicitly,
if \(A\subset\mathcal V_{L,r}\), \(|A|=k\), then
\begin{equation}
 (H_{L,k}^{(r)}f)(A)
 =\sum_{\substack{\xi\in A,\ \eta\notin A}}
 c_{\xi\eta}\bigl(f(A)-f(A\setminus\{\xi\}\cup\{\eta\})\bigr).
 \label{eq:decorated-exclusion-generator}
\end{equation}
Since \(\widetilde H_L=(r/2)\mathcal H_L^{(r)}\),
\begin{equation}
 \Tr e^{-\beta\widetilde H_L}
 =\Tr e^{-t\mathcal H_L^{(r)}},
 \qquad t:=\frac{\beta r}{2}=\beta S.
 \label{eq:decorated-time-normalization}
\end{equation}

We now compute the one-particle operator without suppressing its
normalization.  Put
\(u_r=r^{-1/2}(1,\ldots,1)\in\mathbb C^r\),
\(P=|u_r\rangle\langle u_r|\), and \(Q=I-P\).  From
\eqref{eq:decorated-conductances}, for
\(\psi\in\ell^2(B_L)\otimes\mathbb C^r\),
\begin{align}
 (h_r\psi)(x,a)
 &=\sum_{y\sim x}\sum_{b=1}^r\frac1r
   \bigl(\psi(x,a)-\psi(y,b)\bigr)\notag\\
 &=\deg_{B_L}(x)\psi(x,a)
   -\sum_{y\sim x}(P\psi)(y,a).
 \label{eq:decorated-one-particle-action}
\end{align}
It follows that
\begin{equation}
 h_r=(-\Delta_{\mathrm{bc}})\otimes P
     +M_{\deg,\mathrm{bc}}\otimes Q.
 \label{eq:acoustic-optical-split}
\end{equation}
Here \(\mathrm{bc}\) denotes the chosen free or periodic boundary
condition, and
\((M_{\deg,\mathrm{bc}}u)(x)=\deg_{B_L}(x)u(x)\).
The spaces \(\ell^2(B_L)\otimes P\mathbb C^r\) and
\(\ell^2(B_L)\otimes Q\mathbb C^r\) are invariant under \(h_r\).  On the
first, \(h_r\) is unitarily equivalent to the scalar base-square
Laplacian.  On the second,
\begin{equation}
 h_r\geq \min_{x\in B_L}\deg_{B_L}(x)\,Q\geq2Q,
 \label{eq:optical-gap-two}
\end{equation}
for \(L\geq3\), for both boundary conditions.  Thus the normalized
one-particle energies are \(\varepsilon(p)\) on the \(P\)-subspace and at
least \(2\) on the \(Q\)-subspace.  In physical units they are
\(S\varepsilon(p)\) and at least \(2S\), respectively.

\subsection{Collision control at fixed fibre width}
\label{subsec:decorated-collision}

Let
\[
 X_{L,k}^{(r)}=\mathcal V_{L,r}^k,
 \qquad
 D_{L,k}^{(r)}
 =\{(\xi_1,\ldots,\xi_k):\xi_i\ne\xi_j\text{ for }i\ne j\}.
\]
On the ordered product space put
\[
 L_{0,L,k}^{(r)}:=\sum_{i=1}^kh_{r,i}.
\]
Here \(h_{r,i}\) acts as \(h_r\) in coordinate \(i\) and as the identity
in all other coordinates.  Denote by
\(\mathcal E(\Gamma_{L,r}^{\square k})\) the set of unoriented edges of
the weighted Cartesian product graph.
Every product edge changes a single coordinate along an edge of
\(\Gamma_{L,r}\), and is assigned the conductance of that decorated edge.
Thus, with each unoriented product edge counted once,
\begin{equation}
 \langle F,L_{0,L,k}^{(r)}F\rangle
 =\sum_{e\in\mathcal E(\Gamma_{L,r}^{\square k})}
 c(e)|\nabla_eF|^2.
 \label{eq:decorated-product-form}
\end{equation}
The restriction of this form to \(D_{L,k}^{(r)}\), after passing to the
permutation-invariant subspace, is precisely
\eqref{eq:decorated-exclusion-generator}.  For \(k=1\) there is no
collision set, and we put \(E_{L,1}^{(r)}=I\) and
\(N_{L,1}^{(r)}=0\).  Let \(2\leq k\leq rV/2\) below.  Decompose the
symmetric product space into the distinct and collision sets.  The
Dirichlet collision block is strictly positive, and
the block decomposition, harmonic extension, and Schur complement of
Section~\ref{subsec:harmonic-extension} give operators
\(E_{L,k}^{(r)}\) and \(N_{L,k}^{(r)}\geq0\).  In particular,
\begin{align}
 L_{0,L,k}^{(r)}E_{L,k}^{(r)}f
 &=\binom{(H_{L,k}^{(r)}+N_{L,k}^{(r)})f}{0},
 \label{eq:decorated-L-on-extension}\\
 0\leq N_{L,k}^{(r)}&\leq R_{L,k}^{(r)},
 \label{eq:decorated-N-below-R}
\end{align}
where \(R_{L,k}^{(r)}\) is multiplication by the total conductance of
jumps from a distinct configuration into the collision set.
To justify the asserted positivity: \(\Gamma_{L,r}\) and hence its
Cartesian product are connected, while \(D_{L,k}^{(r)}\ne\varnothing\)
for \(k\leq rV\).  The Dirichlet restriction of the Laplacian of a
connected finite graph to the complement of a nonempty set has trivial
kernel; the same is true after restriction to the symmetric subspace.

For a product edge call it bad if at least one endpoint is outside
\(D_{L,k}^{(r)}\), and set
\begin{equation}
 q_{\mathrm{bad}}^{(r)}(F)
 :=\sum_{e\ \mathrm{bad}}c(e)|\nabla_eF|^2.
 \label{eq:decorated-bad-edge-form}
\end{equation}
For \(i<j\), let \(\mathcal C_{ij}=\{\xi_i=\xi_j\}\), and let
\(q_{ij}^{(r)}\) be the part of the weighted product form carried by edges
with at least one endpoint in \(\mathcal C_{ij}\).  Every bad edge is
contained in at least one of these pair forms, whence
\begin{equation}
 q_{\mathrm{bad}}^{(r)}(F)
 \leq\sum_{1\leq i<j\leq k}q_{ij}^{(r)}(F).
 \label{eq:decorated-bad-pair-cover}
\end{equation}

We first prove the required single-pair estimate.  Set
\(\Pi_{\sigma\tau}^{ij}=\sigma_i\tau_j\), for
\(\sigma,\tau\in\{P,Q\}\), and
\(F_{\sigma\tau}=\Pi_{\sigma\tau}^{ij}F\).  Since a squared difference
is a quadratic form and \(F=\sum_{\sigma,\tau}F_{\sigma\tau}\),
edgewise Cauchy--Schwarz gives
\begin{equation}
 q_{ij}^{(r)}(F)
 \leq4\sum_{\sigma,\tau\in\{P,Q\}}
 q_{ij}^{(r)}(F_{\sigma\tau}).
 \label{eq:decorated-channel-cauchy}
\end{equation}

We first estimate the \(PP\) component, including free boundary conditions.
After identifying \(P\mathbb C^r\) with \(\mathbb C\), the function
\(F_{PP}\) may be written as a vector-valued function
\[
 G(x_i,x_j)\in\mathcal K_{\mathrm{sp}},
\]
where \(\mathcal K_{\mathrm{sp}}\) contains the variables of the remaining
\(k-2\) particles.  More explicitly,
\begin{equation}
 F_{PP}(x,a;y,b;\widehat\xi)=\frac1rG(x,y;\widehat\xi).
 \label{eq:PP-unitary-identification}
\end{equation}
Because \(F\) is permutation invariant, \(G\) is exchange symmetric in
the two selected base coordinates.  Let
\[
 A_{\mathrm{sp}}:=\sum_{\ell\ne i,j}h_{r,\ell}\geq0
\]
on \(\mathcal K_{\mathrm{sp}}\).  The first positive eigenvalue of an
individual \(P\)-subspace factor is at least \(cL^{-2}\), and every
\(Q\)-subspace eigenvalue is at least \(2\).  Consequently, for a constant
\(c_r>0\),
\begin{equation}
 \operatorname{spec}(A_{\mathrm{sp}})
 \subset\{0\}\cup[c_rL^{-2},\infty).
 \label{eq:spectator-gap}
\end{equation}
The free product operator on this component is
\[
 (-\Delta_{\mathrm{bc}})_i+(-\Delta_{\mathrm{bc}})_j+A_{\mathrm{sp}}.
\]
For a normal edge, summation over its two fibre labels and use of the
conductance \(1/r\) gives
\[
 \sum_{a,b=1}^r\frac1r\left|\frac1r
 (G(y,x;\widehat\xi)-G(x,x;\widehat\xi))\right|^2
 =\frac1r|G(y,x;\widehat\xi)-G(x,x;\widehat\xi)|^2.
\]
On the pair diagonal the two selected fibre labels coincide.  For a
tangential spectator edge, summing this common label therefore contributes
the factor \(r\cdot r^{-2}=r^{-1}\).  Thus
\begin{equation}
 q_{ij}^{(r)}(F_{PP})=\frac1r q_{ij}^{\mathrm{base}}(G),
 \label{eq:PP-form-reduction}
\end{equation}
where the nonnegative spectator operator is retained in
\(q_{ij}^{\mathrm{base}}\).

For periodic boundary conditions, resolve \(A_{\mathrm{sp}}\) spectrally.
On an eigenspace of \(A_{\mathrm{sp}}\), its value is a number
\(a\in\{0\}\cup[c_rL^{-2},\infty)\) in
\eqref{eq:product-fibre-energy}.  In precisely this range, the two lattice sums
\eqref{eq:tangential-restriction-sum} and
\eqref{eq:normal-restriction-sum} are uniform in \(a\).  Integrating the
scalar estimate over the spectral measure of \(A_{\mathrm{sp}}\) therefore
gives
\begin{equation}
 q_{ij}^{(r)}(F_{PP})
 \leq\frac Cr\|L_{0,L,k}^{(r)}F_{PP}\|^2.
 \label{eq:decorated-PP-periodic-trace}
\end{equation}

For free boundary conditions we reflect only the two selected spatial
coordinates and leave \(\mathcal K_{\mathrm{sp}}\) and \(A_{\mathrm{sp}}\) unchanged.
Let \(J\) be the two-dimensional even extension defined in
\eqref{eq:neumann-reflection-intertwining}, and let
\(\mathcal J_{ij}\) act as \(J\) in coordinates \(i,j\) and as the
identity on \(\mathcal K_{\mathrm{sp}}\).  Then
\begin{align}
 \mathcal J_{ij}
 \bigl[(-\Delta_N)_i+(-\Delta_N)_j+A_{\mathrm{sp}}\bigr]
 &=\bigl[(-\Delta_{\mathrm{per},2L})_i
 +(-\Delta_{\mathrm{per},2L})_j+A_{\mathrm{sp}}\bigr]\mathcal J_{ij},
 \label{eq:decorated-PP-reflection-intertwining}\\
 \|\mathcal J_{ij}G\|^2&=16\|G\|^2.
 \label{eq:decorated-PP-reflection-norm}
\end{align}
The two selected coordinates share one reflected spatial variable on the
pair diagonal.  Hence the diagonal form is multiplied by \(4\), whereas
the squared product norm on the right of the trace estimate is multiplied
by \(16\).  Applying the preceding periodic estimate on each eigenspace
of \(A_{\mathrm{sp}}\) to \(\mathcal J_{ij}G\) proves, for both boundary
conditions,
\begin{equation}
 q_{ij}^{(r)}(F_{PP})
 \leq C_r\|L_{0,L,k}^{(r)}F_{PP}\|^2.
 \label{eq:decorated-PP-trace}
\end{equation}

The other three components do not require a trace estimate.  The
projections \(P_i,Q_i\) commute with \(L_{0,L,k}^{(r)}\).  In each of the
components \(PQ,QP,QQ\), at least one selected fibre component lies in
\(Q\), and hence \(L_{0,L,k}^{(r)}\geq2\).  Since \(q_{ij}^{(r)}\) is a
subform of the
full product form,
\begin{equation}
 q_{ij}^{(r)}(F_{\sigma\tau})
 \leq\langle F_{\sigma\tau},L_{0,L,k}^{(r)}F_{\sigma\tau}\rangle
 \leq\frac12\|L_{0,L,k}^{(r)}F_{\sigma\tau}\|^2
 \label{eq:decorated-optical-trace}
\end{equation}
for \((\sigma,\tau)\in\{(P,Q),(Q,P),(Q,Q)\}\).  The four projections are
mutually orthogonal and commute with \(L_{0,L,k}^{(r)}\), so
\[
 \sum_{\sigma,\tau\in\{P,Q\}}
 \|L_{0,L,k}^{(r)}F_{\sigma\tau}\|^2
 =\|L_{0,L,k}^{(r)}F\|^2.
\]
Together with \eqref{eq:decorated-channel-cauchy}, this proves
\begin{equation}
 q_{ij}^{(r)}(F)\leq C_r\|L_{0,L,k}^{(r)}F\|^2.
 \label{eq:decorated-single-pair-trace}
\end{equation}
Summing over the fewer than \(k^2/2\) pairs gives
\begin{equation}
 q_{\mathrm{bad}}^{(r)}(F)
 \leq C_rk^2\|L_{0,L,k}^{(r)}F\|^2.
 \label{eq:decorated-all-collision-strata}
\end{equation}

The decorated block multiplication in
\eqref{eq:decorated-L-on-extension} shows that
\[
 \|L_{0,L,k}^{(r)}E_{L,k}^{(r)}f\|^2
 =\|(H_{L,k}^{(r)}+N_{L,k}^{(r)})f\|^2.
\]
Splitting the decorated product form into edges contained in
\(D_{L,k}^{(r)}\) and edges having a collision endpoint gives
\begin{align*}
 \langle E_{L,k}^{(r)}f,
 L_{0,L,k}^{(r)}E_{L,k}^{(r)}f\rangle
 &=\langle f,H_{L,k}^{(r)}f\rangle
 +q_{\mathrm{bad}}^{(r)}(E_{L,k}^{(r)}f),\\
 &=\langle f,(H_{L,k}^{(r)}+N_{L,k}^{(r)})f\rangle.
\end{align*}
Therefore
\begin{equation}
 q_{\mathrm{bad}}^{(r)}(E_{L,k}^{(r)}f)
 =\langle f,N_{L,k}^{(r)}f\rangle.
 \label{eq:decorated-N-as-bad-energy}
\end{equation}
Applying \eqref{eq:decorated-all-collision-strata} to the harmonic
extension proves
\begin{equation}
 N_{L,k}^{(r)}
 \leq C_rk^2(H_{L,k}^{(r)}+N_{L,k}^{(r)})^2.
 \label{eq:decorated-quadratic-prebound}
\end{equation}

We next supply the first-order estimate, including the deletion required
by the moving-particle argument.

\begin{lemma}
\label{lem:decorated-punctured-resistance}
For every fixed \(r\in\mathbb N\), there is \(C_r<\infty\) such that, for
either choice of boundary conditions, every integer \(L\geq3\), every
\(\alpha\in\mathcal V_{L,r}\), and all
\(\xi,\eta\in\mathcal V_{L,r}\setminus\{\alpha\}\),
\begin{equation}
 \mathsf R_{\mathrm{eff}}^{\Gamma_{L,r}\setminus\{\alpha\}}
 (\xi,\eta)\leq C_r\log(2L).
 \label{eq:decorated-punctured-resistance}
\end{equation}
\end{lemma}

\begin{proof}
The case \(\xi=\eta\) is trivial.  Otherwise write
\(\alpha=(z,c)\), \(\xi=(x,a)\), and \(\eta=(y,b)\).  If \(x\ne z\),
choose a neighbour \(x_*\sim x\) with \(x_*\ne z\); such a neighbour
exists because every vertex of either square has at least two neighbours.
If \(x=z\), choose any \(x_*\sim z\).  If \(y\ne z\), choose a neighbour
\(y_*\sim y\) with \(y_*\ne z\); if \(y=z\), choose any \(y_*\sim z\).
Send current \(1/r\) from \((x,a)\) to each \((x_*,d)\), \(d\in[r]\).
The resulting endpoint flow has unit divergence at \(\xi\), divergence
\(-1/r\) at every
\((x_*,d)\), avoids \(\alpha\), and has energy
\begin{equation}
 r\,\frac{(1/r)^2}{1/r}=1.
 \label{eq:decorated-fan-energy}
\end{equation}
At \(\eta\), use the corresponding reversed endpoint flow.

By Lemma~\ref{lem:punctured-resistance}, there is a unit flow \(j\) from
\(x_*\) to \(y_*\) in \(B_L\setminus\{z\}\) with energy at most
\(C\log(2L)\).  Lift it to the decorated graph by putting
\[
 J((u,d),(v,e)):=\frac1{r^2}j(u,v),
 \qquad d,e\in[r],\quad u,v\ne z.
\]
At every decorated vertex over \(u\), the divergence of \(J\) is
\(r^{-1}\operatorname{div}j(u)\), so it connects the two uniform fibre
distributions produced by the endpoint flows.  Moreover,
\begin{equation}
 \sum_{\{u,v\}}\sum_{d,e=1}^r
 \frac{|j(u,v)/r^2|^2}{1/r}
 =\frac1r\sum_{\{u,v\}}|j(u,v)|^2
 \leq\frac Cr\log(2L).
 \label{eq:decorated-lifted-flow-energy}
\end{equation}
The sum of the two endpoint flows and the lifted flow is a unit flow from \(\xi\)
to \(\eta\) avoiding \(\alpha\).  If some of its supports overlap, the
inequality \(|u+v+w|^2\leq3(|u|^2+|v|^2+|w|^2)\) changes only the constant.
Thomson's principle proves
\eqref{eq:decorated-punctured-resistance}.  The case \(x_*=y_*\) uses the
zero base flow and is included.
\end{proof}

For later use, let \(M:=rV=|\mathcal V_{L,r}|\), set
\[
 \Omega_{L,k}^{(r)}:=
 \{A\subset\mathcal V_{L,r}:|A|=k\},
\]
let \(\nu_{L,k}^{(r)}\) be the uniform probability measure on this set,
and put
\begin{equation}
 R_{L,k}^{(r)}(A)
 :=\sum_{\substack{\xi,\eta\in A}}c_{\xi\eta}.
 \label{eq:decorated-contact-multiplier}
\end{equation}
The multiplication operator \(R_{L,k}^{(r)}\) in
\eqref{eq:decorated-N-below-R} is therefore given by
\eqref{eq:decorated-contact-multiplier}; the sum is oriented, as are the
attempted jumps into occupied sites.

\begin{lemma}
\label{lem:decorated-contact-bound}
For every fixed \(r\in\mathbb N\), there is \(C_r<\infty\) such that,
for either choice of boundary conditions, every integer \(L\geq3\), every
integer \(1\leq k\leq\lfloor M/2\rfloor\), and every
\(f\in\ell^2(\Omega_{L,k}^{(r)},\nu_{L,k}^{(r)})\),
\begin{equation}
 \langle f,R_{L,k}^{(r)}f\rangle_{\nu_{L,k}^{(r)}}
 \leq C_r\frac{k^2}{M}\|f\|_{\nu_{L,k}^{(r)}}^2
 +C_rk\log(2L)
 \langle f,H_{L,k}^{(r)}f\rangle_{\nu_{L,k}^{(r)}}.
 \label{eq:decorated-contact-bound}
\end{equation}
\end{lemma}

\begin{proof}
The proof is included to keep track of the vertex number \(M=rV\) and the
weighted degree.  For \(k=1\), one has \(R_{L,1}^{(r)}=0\), so the assertion
holds.  Hence assume \(k\geq2\).  Fix an occupied ordered edge
\((\alpha,\zeta)\), remove
\(\zeta\), put \(B=A\setminus\{\zeta\}\), and set
\(n=M-k+1\).  With
\[
 \bar f_B:=\frac1n\sum_{y\notin B}f(B\cup\{y\}),
\]
Jensen gives
\begin{equation}
 |f(A)|^2\leq2|\bar f_B|^2
 +\frac2n\sum_{y\notin B}
 |f(A)-f(B\cup\{y\})|^2.
 \label{eq:decorated-contact-jensen}
\end{equation}
For the mean term, rewrite the sum first over the \((k-1)\)-set \(B\).
Using \eqref{eq:decorated-weighted-degree} and Jensen once more gives
\begin{align}
 &\frac1{\binom Mk}\sum_{|A|=k}
 \sum_{\alpha,\zeta\in A}c_{\alpha\zeta}
 |\bar f_{A\setminus\{\zeta\}}|^2\notag\\
 &\quad\leq
 \frac{\Delta_c(k-1)}{\binom Mk}
 \sum_{|B|=k-1}\frac1n\sum_{y\notin B}|f(B\cup\{y\})|^2\notag\\
 &\quad=
 \frac{\Delta_c k(k-1)}n\|f\|_{\nu_{L,k}^{(r)}}^2
 \leq C\frac{k^2}{M}\|f\|_{\nu_{L,k}^{(r)}}^2,
 \label{eq:decorated-contact-mean}
\end{align}
The final inequality in \eqref{eq:decorated-contact-mean} uses
\(k\leq M/2\).

For the fluctuation term condition on the occupied anchor \(\alpha\).  On
the punctured weighted graph
\(\Gamma_{L,r}^{\alpha}:=\Gamma_{L,r}\setminus\{\alpha\}\), let
\(f_\alpha(\omega)=f(\omega\cup\{\alpha\})\), where
\(|\omega|=k-1\).  Let \(\nu^\alpha\) be the uniform law on the
\((k-1)\)-subsets of
\(\mathcal V_{L,r}\setminus\{\alpha\}\), and set
\[
 \mathcal E_\alpha(f_\alpha):=
 \langle f_\alpha,
 H_{\Gamma_{L,r}^{\alpha},k-1}f_\alpha\rangle_{\nu^\alpha}.
\]
For \(\zeta,y\ne\alpha\), let \(\omega^{\zeta y}\) denote the
configuration obtained by exchanging the occupations of \(\zeta\) and
\(y\).  The weighted moving-particle lemma \cite{Chen2017} gives
\begin{equation}
 \frac12\mathbb E_{\nu^\alpha}
 |f_\alpha(\omega^{\zeta y})-f_\alpha(\omega)|^2
 \leq
 \mathsf R_{\mathrm{eff}}^{\Gamma_{L,r}^{\alpha}}(\zeta,y)
 \mathcal E_\alpha(f_\alpha).
 \label{eq:decorated-moving-particle}
\end{equation}
The left-hand side of \eqref{eq:decorated-moving-particle} is understood as
zero unless exactly one of \(\zeta,y\) is occupied.  Chen states the
estimate for a Bernoulli product measure.  To obtain
\eqref{eq:decorated-moving-particle} at fixed particle number, extend
\(f_\alpha\) by zero off
the \((k-1)\)-particle sector.  Exchanges preserve particle number, and
every configuration in that sector has the same Bernoulli weight
\(p^{k-1}(1-p)^{M-k}\); this common factor occurs on both sides and
cancels.  Thus there is no hidden dependence on \(M\), \(k\), or \(p\).

More explicitly, after conditioning on \(\alpha\in A\), the
fluctuation contribution (apart from the factor \(2\) in
\eqref{eq:decorated-contact-jensen}) is
\begin{align}
 F={}&\frac{k}{Mn}
 \sum_{\alpha\in\mathcal V_{L,r}}
 \sum_{\zeta\ne\alpha}c_{\alpha\zeta}
 \sum_{y\ne\alpha}
 \mathbb E_{\nu^\alpha}\!\left[
 \1_{\{\zeta\in\omega,\,y\notin\omega\}}
 |f_\alpha(\omega\setminus\{\zeta\}\cup\{y\})
     -f_\alpha(\omega)|^2\right].
 \label{eq:decorated-conditioned-fluctuation}
\end{align}
Apply \eqref{eq:decorated-moving-particle}, then use
\((M-1)/n\leq2\), \eqref{eq:decorated-weighted-degree}, and
Lemma~\ref{lem:decorated-punctured-resistance}.  This gives
\begin{equation}
 F\leq C_r\frac{k\log(2L)}M
 \sum_{\alpha\in\mathcal V_{L,r}}\mathcal E_\alpha(f_\alpha).
 \label{eq:decorated-fluctuation-before-anchor}
\end{equation}
Finally use the exact anchored-energy identity
\begin{equation}
 \frac{k}{M}\sum_{\alpha\in\mathcal V_{L,r}}
 \mathcal E_\alpha(f_\alpha)
 =(k-1)\langle f,H_{L,k}^{(r)}f\rangle_{\nu_{L,k}^{(r)}}.
 \label{eq:decorated-anchored-energy}
\end{equation}
Indeed, every allowed exclusion move has precisely \(k-1\) occupied
anchors outside its two endpoints.  The fluctuation term is therefore at
most \(C_rk\log(2L)\langle f,H_{L,k}^{(r)}f\rangle\).  Together with
\eqref{eq:decorated-contact-mean}, this proves
\eqref{eq:decorated-contact-bound}.
\end{proof}

We also record the gap with its normalization.  By
\eqref{eq:acoustic-optical-split}, the nonzero spectrum of \(h_r\) is the
union of the positive spectrum of the base-square Laplacian and the
spectrum of \(M_{\deg,\mathrm{bc}}\) repeated \(r-1\) times.  Hence
\begin{equation}
 \gap(h_r)\geq cL^{-2}.
 \label{eq:decorated-one-particle-gap}
\end{equation}
The weighted graph is connected.  Aldous' theorem applies to the weighted
interchange process, in which the labels at the endpoints of an edge are
exchanged at the edge's conductance rate
\cite{CaputoLiggettRichthammer2010}.  Forgetting the labels and retaining
only which vertices carry one of \(k\) selected labels gives the
\(k\)-particle exclusion process; hence its gap is at least
\(\gap(h_r)\).  Conversely,
if \(h_r\phi=\gap(h_r)\phi\) with \(\sum_\xi\phi(\xi)=0\), then
\[
 F_k(A):=\sum_{\xi\in A}\phi(\xi)
\]
is nonzero for \(1\leq k\leq M-1\).  Indeed, if \(F_k\) vanished
identically, comparing two \(k\)-sets which differ only by replacing
\(\xi\) with \(\eta\) would give \(\phi(\xi)=\phi(\eta)\) for every
\(\xi,\eta\); then \(\phi\) would be constant and hence zero.  Moreover,
using the symmetry \(c_{\xi\eta}=c_{\eta\xi}\),
\begin{align*}
 (H_{L,k}^{(r)}F_k)(A)
 &=\sum_{\substack{\xi\in A\\ \eta\notin A}}
 c_{\xi\eta}\bigl(\phi(\xi)-\phi(\eta)\bigr)\\
 &=\sum_{\xi\in A}\sum_{\eta\in\mathcal V_{L,r}}
 c_{\xi\eta}\bigl(\phi(\xi)-\phi(\eta)\bigr)\\
 &=\gap(h_r)\sum_{\xi\in A}\phi(\xi)
 =\gap(h_r)F_k(A).
\end{align*}
In the second line the terms with \(\xi,\eta\in A\) cancel in pairs.
Thus the reverse gap inequality also holds, and
\begin{equation}
 \gap(H_{L,k}^{(r)})=\gap(h_r)\geq cL^{-2}.
 \label{eq:decorated-exclusion-gap}
\end{equation}
Both \(H_{L,k}^{(r)}\) and \(N_{L,k}^{(r)}\) annihilate constants.  On the
orthogonal complement of constants,
\eqref{eq:decorated-N-below-R},
Lemma~\ref{lem:decorated-contact-bound}, and
\eqref{eq:decorated-exclusion-gap} give
\begin{align}
 N_{L,k}^{(r)}
 &\leq C_r\left(\frac{k^2}{rV}L^2
       +k\log(2L)\right)H_{L,k}^{(r)}\notag\\
 &\leq b_{L,k}^{(r)}H_{L,k}^{(r)},
 \qquad b_{L,k}^{(r)}:=C_rk^2\log(2L).
 \label{eq:decorated-first-order-bound}
\end{align}

\begin{proposition}
\label{prop:decorated-normal-trace}
For every fixed \(r\in\mathbb N\), there is \(C_{\mathrm{tr},r}<\infty\)
such that, for every integer \(L\geq3\), for either choice of boundary
conditions, and for every integer
\(1\leq k\leq\lfloor rV/2\rfloor\),
\begin{equation}
 0\leq N_{L,k}^{(r)}
 \leq C_{\mathrm{tr},r}k^6\log^2(2L)
 (H_{L,k}^{(r)})^2
 \label{eq:decorated-normal-trace}
\end{equation}
as an operator on \(\ell^2_{\mathrm{Sym}}(D_{L,k}^{(r)})\).
\end{proposition}

\begin{proof}
For \(k=1\), one has \(N_{L,1}^{(r)}=0\), so the assertion holds.  Assume
\(k\geq2\).
Work first on the orthogonal complement of constants and set
\(\mathcal K=H^{-1/2}NH^{-1/2}\), with
\(H=H_{L,k}^{(r)}\) and \(N=N_{L,k}^{(r)}\).  By
\eqref{eq:decorated-first-order-bound},
\(0\leq\mathcal K\leq b_{L,k}^{(r)}I\).  Multiplying
\eqref{eq:decorated-quadratic-prebound} on the left and right by
\(H^{-1/2}\) gives
\[
 \mathcal K\leq C_rk^2(I+\mathcal K)H(I+\mathcal K).
\]
Since \(\mathcal K\) commutes with \(I+\mathcal K\), multiplying on the
left and right by \((I+\mathcal K)^{-1}\) yields
\begin{equation}
 \mathcal K(I+\mathcal K)^{-2}\leq C_rk^2H.
 \label{eq:decorated-absorbed-calculus}
\end{equation}
For \(0\leq x\leq b_{L,k}^{(r)}\),
\[
 x\leq(1+b_{L,k}^{(r)})^2\frac{x}{(1+x)^2}.
\]
Applying this scalar inequality to the spectral decomposition of
\(\mathcal K\) in \eqref{eq:decorated-absorbed-calculus}, and then
multiplying on the left and right by \(H^{1/2}\), proves
\begin{align*}
 N&\leq C_rk^2(1+b_{L,k}^{(r)})^2H^2\\
  &\leq C_{\mathrm{tr},r}k^6\log^2(2L)H^2.
\end{align*}
Since \(H_{L,k}^{(r)}\) and \(N_{L,k}^{(r)}\) annihilate constants,
\eqref{eq:decorated-normal-trace} holds on the full exclusion space.
\end{proof}

\subsection{Completion of the free-energy asymptotics and a remainder}
\label{subsec:completion-fixed-spin}

Besides completing the proof of the main theorem, this subsection records
the quantitative remainder announced in the introduction.

\begin{corollary}[A quantitative remainder]
\label{cor:quantitative-free-energy}
For every fixed \(S\in\frac12\N\), there are constants
\(C_S<\infty\) and \(t_S<\infty\) such that, whenever
\(t:=\beta S\geq t_S\),
\begin{equation}
 \left|\beta^2S f_2(\beta,S)+\frac{\pi}{24}\right|
 \leq C_S t^{-1/27}(\log t)^{16}.
 \label{eq:quantitative-free-energy}
\end{equation}
\end{corollary}

We return to free boundary conditions.  We first give the comparison used
both for the rough trace and for the magnon-number cutoff.  If \(m\geq2\),
write \(v=m^2\) and \(M_m=rv\).  Let \(K_{M_m,k}\) be the unit-conductance
exclusion generator of the complete graph on the same \(M_m\) decorated
vertices.
Let \(\mathcal M_{m,k}^{(r)}\) be the global highest-weight multiplicity
space of total spin \(M_m/2-k\), and let
\(H_{m,k}^{(r),\mathrm{HW}}\) denote the restriction of
\(H_{m,k}^{(r)}\) to this space.

\begin{lemma}
\label{lem:decorated-rough-sector}
For every fixed \(r\in\mathbb N\), there is \(c_{\mathrm{sec},r}>0\) such
that, for every integer \(m\geq2\) and every integer
\(1\leq k\leq\lfloor M_m/2\rfloor\),
\begin{equation}
 H_{m,k}^{(r),\mathrm{HW}}
 \geq c_{\mathrm{sec},r}\frac{k}{m^2}.
 \label{eq:decorated-rough-sector}
\end{equation}
Consequently, for every \(s>0\), if
\(Z_{m,r}(s):=\Tr\exp(-s\mathcal H_m^{(r)})\), then
\begin{equation}
 Z_{m,r}(s)
 \leq(M_m+1)\exp\left\{M_m
 e^{-c_{\mathrm{sec},r}s/m^2}\right\}.
 \label{eq:decorated-rough-trace}
\end{equation}
\end{lemma}

\begin{proof}
For each ordered pair of decorated vertices
\(\xi=(x,a)\), \(\eta=(y,b)\), choose a path as follows.  If \(x\ne y\),
use the base-square Manhattan path which first moves horizontally and then
vertically; use the label \(1\) at all intermediate base vertices and the
prescribed labels at the endpoints.  If \(x=y\) and \(a\ne b\), use the
two-edge path through a fixed neighbour of \(x\).  Every such path has
length at most \(2m+2\); denote it by \(\gamma_{\xi\eta}\).

Let \(\tau_{\xi\eta}\) exchange the occupations at its endpoints and
\(\tau_e\) exchange occupations along a decorated edge.  Moving one
endpoint along the path and then moving the other endpoint back gives a
word of at most \(2\ell-1\) adjacent exchanges for
\(\tau_{\xi\eta}\).  Telescoping and Cauchy--Schwarz therefore give
\begin{equation}
 \|f-\tau_{\xi\eta}f\|^2
 \leq(2\ell-1)\sum_{e\in\gamma_{\xi\eta}}2
 \|f-\tau_ef\|^2.
 \label{eq:decorated-path-telescoping}
\end{equation}
For a fixed horizontal base edge, at most \(Cm^3\) ordered pairs of base
vertices have a horizontal-first Manhattan path using that edge: the two
horizontal endpoints contribute \(O(m^2)\), and the unused vertical
endpoint contributes \(O(m)\).  The same count holds for a vertical edge.
The endpoint labels contribute at most the fixed factor \(r^2\), and the
two-edge same-fibre paths contribute only \(O_r(1)\).  Multiplication by
the path length shows that the path-length-weighted load of every
decorated edge is at most \(C_rm^4\).  Since the local edge conductance in
\(H_{m,k}^{(r)}\) is \(1/r\), summing
\eqref{eq:decorated-path-telescoping} gives
\begin{equation}
 K_{M_m,k}\leq C_rm^4H_{m,k}^{(r)}.
 \label{eq:decorated-complete-path-comparison}
\end{equation}

The decorated system contains \(M_m=rm^2\) spin-half vertices.  Applying
\eqref{eq:complete-graph-casimir-sector} with \(n=M_m\) shows that, on
the highest-weight multiplicity space of total spin \(M_m/2-k\),
\begin{equation}
 K_{M_m,k}^{\mathrm{HW}}=k(M_m-k+1)I
 \geq\frac12kM_mI.
 \label{eq:decorated-complete-spin-identity}
\end{equation}
Combining \eqref{eq:decorated-complete-path-comparison} and
\eqref{eq:decorated-complete-spin-identity}, and using \(M_m=rm^2\), proves
\eqref{eq:decorated-rough-sector}.

Finally decompose the full \(M_m\)-spin-half space into global
\(\mathrm{SU}(2)\) multiplets.  Their dimensions are at most \(M_m+1\),
and the multiplicity of total spin \(M_m/2-k\) is at most
\(\binom{M_m}{k}\).  Therefore
\begin{align*}
 Z_{m,r}(s)
 &\leq(M_m+1)\sum_{k=0}^{\lfloor M_m/2\rfloor}
 \binom{M_m}{k}e^{-c_{\mathrm{sec},r}sk/m^2}\\
 &\leq(M_m+1)
 (1+e^{-c_{\mathrm{sec},r}s/m^2})^{M_m}\\
 &\leq(M_m+1)\exp\{M_m e^{-c_{\mathrm{sec},r}s/m^2}\}.
\end{align*}
\end{proof}

\begin{proof}[Proof of Corollary~\ref{cor:quantitative-free-energy}]
Together with the known inequality recalled at the end, this proof also
completes Theorem~\ref{thm:main-free-energy}.
For the quantitative pressure estimate put \(t:=\beta S\) and choose the
box-enlargement parameter \(g:=t^a\), where \(0<a<1/13\).  Fix a
sufficiently large constant \(A_r>0\), depending only on \(r=2S\).
Choose integers \(m,L\), with \(L\) a multiple of \(m\), so that
\begin{equation}
 m^2\asymp \frac{t}{A_r\log t},
 \qquad V:=L^2\asymp tg.
 \label{eq:quantitative-rounded-scales}
\end{equation}
The constants implicit in \(\asymp\), here and below, are independent of
\(t\).  Such a choice is obtained by first choosing \(m\) and then taking
\(L/m\) to be the nearest integer to
\((A_rg\log t)^{1/2}\).  Increasing \(A_r\), if necessary, absorbs the
integer roundings.

Apply Lemma~\ref{lem:decorated-rough-sector} in every \(m\)-cell.
Dropping the positive bonds between cells and using
\eqref{eq:decorated-rough-trace} at time \(t/2\) gives
\begin{align}
 \log\Tr e^{-(t/2)\mathcal H_L^{(r)}}
 &\leq\frac{V}{m^2}
 \left[\log(rm^2+1)+rm^2
 e^{-c_{\mathrm{sec},r}t/(2m^2)}\right]\notag\\
 &\leq C_rg(\log t)^2.
 \label{eq:quantitative-rough-trace}
\end{align}
Indeed, after increasing \(A_r\) by a fixed \(r\)-dependent factor, the
rounding in \eqref{eq:quantitative-rounded-scales} still leaves
\(e^{-c_{\mathrm{sec},r}t/(2m^2)}\leq t^{-3}\).  Since
\(V/m^2=O_r(g\log t)\), the logarithmic term in
\eqref{eq:quantitative-rough-trace} is \(O_r(g(\log t)^2)\), and the
exponential term is \(O_r(Vt^{-3})=O_r(gt^{-2})\).

Set
\begin{equation}
 e_{0,r}:=C_{\mathrm{cut},r}\frac{g(\log t)^2}{t},
 \qquad
 K_r(t):=\left\lfloor c_{\mathrm{sec},r}^{-1}Ve_{0,r}\right\rfloor.
 \label{eq:quantitative-cutoffs}
\end{equation}
where \(C_{\mathrm{cut},r}\) will be chosen below.  Put \(M=rV\), and
write
\[
 \widetilde Z_{L,r}(u):=\Tr e^{-u\mathcal H_L^{(r)}}
\]
for the partition function of the full decorated spin-half system.  The
physical trace satisfies
\begin{equation}
 Z_L^{(S)}(\beta)\leq\widetilde Z_{L,r}(t)
 \label{eq:physical-below-decorated-trace}
\end{equation}
by \eqref{eq:physical-decorated-trace} and
\eqref{eq:decorated-time-normalization}.

We first make the high-energy reduction precise.  Let
\(\mathcal M_{L,k}^{(r)}\) be the global highest-weight multiplicity space
of total spin \(M/2-k\), and let \(H_{L,k}^{(r),\HW}\) be the restriction
of \(H_{L,k}^{(r)}\) to this space.  Global rotation invariance gives
\begin{equation}
 \widetilde Z_{L,r}(u)
 =\sum_{k=0}^{\lfloor M/2\rfloor}(M-2k+1)
 \Tr_{\mathcal M_{L,k}^{(r)}}e^{-uH_{L,k}^{(r),\HW}}.
 \label{eq:decorated-su2-trace-decomposition}
\end{equation}
Define the highest-weight high-energy trace by
\[
 T_{>e_{0,r}}(t):=
 \sum_{k=0}^{\lfloor M/2\rfloor}
 \Tr_{\mathcal M_{L,k}^{(r)}}
 \1_{(e_{0,r},\infty)}(H_{L,k}^{(r),\HW})
 e^{-tH_{L,k}^{(r),\HW}}.
\]
For every scalar \(E\geq0\),
\[
 \1_{\{E>e_{0,r}\}}e^{-tE}
 \leq e^{-te_{0,r}/2}e^{-tE/2}.
\]
Consequently, \eqref{eq:decorated-su2-trace-decomposition} and
\eqref{eq:quantitative-rough-trace} imply
\begin{align}
 (M+1)T_{>e_{0,r}}(t)
 &\leq(M+1)e^{-te_{0,r}/2}\widetilde Z_{L,r}(t/2)\notag\\
 &\leq(M+1)e^{-c_rg(\log t)^2}
 \label{eq:quantitative-high-energy-tail}
\end{align}
for all sufficiently large \(t\), provided \(C_{\mathrm{cut},r}\) is
larger than twice the constant in
\eqref{eq:quantitative-rough-trace}.  The factor \(M+1\) bounds the
dimensions of the global spin multiplets.

Applying Lemma~\ref{lem:decorated-rough-sector} with \(m=L\) shows that
an eigenvalue of \(H_{L,k}^{(r),\HW}\) can lie in \([0,e_{0,r}]\) only if
\(k\leq K_r(t)\).  Since \(V\asymp tg\),
\begin{equation}
 K_r(t)\leq C_rg^2(\log t)^2.
 \label{eq:quantitative-number-cutoff}
\end{equation}
In particular, \(K_r(t)\leq M/2\) for all sufficiently large \(t\).

Fix \(1\leq k\leq K_r(t)\), and let
\[
 P_{k,e_{0,r}}^{\HW}:=
 \1_{[0,e_{0,r}]}(H_{L,k}^{(r),\HW}),
\]
regarded as a projection in the full \(k\)-particle exclusion space.  No
invariance of its range under \(N_{L,k}^{(r)}\) is used.  On this range,
the spectral theorem gives
\[
 P_{k,e_{0,r}}^{\HW}(H_{L,k}^{(r)})^2P_{k,e_{0,r}}^{\HW}
 \leq e_{0,r}P_{k,e_{0,r}}^{\HW}H_{L,k}^{(r)}P_{k,e_{0,r}}^{\HW}.
\]
Testing Proposition~\ref{prop:decorated-normal-trace} on the same range
therefore yields
\begin{align}
 &P_{k,e_{0,r}}^{\HW}(H_{L,k}^{(r)}+N_{L,k}^{(r)})
 P_{k,e_{0,r}}^{\HW}\notag\\
 &\qquad\leq(1+\delta_{t,r})
 P_{k,e_{0,r}}^{\HW}H_{L,k}^{(r)}P_{k,e_{0,r}}^{\HW},
 \label{eq:decorated-low-energy-relative-bound}
\end{align}
where
\begin{equation}
 \delta_{t,r}:=
 C_{\mathrm{tr},r}K_r(t)^6\log^2(2L)e_{0,r}
 \leq C_r\frac{g^{13}(\log t)^{16}}{t}.
 \label{eq:quantitative-delta}
\end{equation}
Indeed, \eqref{eq:quantitative-number-cutoff} supplies
\(g^{12}(\log t)^{12}\), the factor \(\log^2(2L)\) supplies two
logarithms, and \(e_{0,r}\) supplies \(g(\log t)^2/t\).  Since
\(g=t^a\) and \(a<1/13\), one has \(\delta_{t,r}=o(1)\).

We next perform the min--max comparison.  Let
\[
 \lambda_{k,1}\leq\cdots\leq\lambda_{k,d_k}\leq e_{0,r}
\]
be all eigenvalues of \(H_{L,k}^{(r),\HW}\) in \([0,e_{0,r}]\), counted
with multiplicity, and let \(\mu_{k,j}^{(r)}\) be the increasingly ordered
eigenvalues of \(L_{0,L,k}^{(r)}\) on
\(\Sym^k\ell^2(\mathcal V_{L,r})\).  The harmonic extensions of the first
\(j\) interacting eigenvectors are linearly independent, because
restriction to \(D_{L,k}^{(r)}\) is a left inverse of
\(E_{L,k}^{(r)}\).  Their norms do not decrease, while
\eqref{eq:decorated-L-on-extension} and
\eqref{eq:decorated-low-energy-relative-bound} bound their Rayleigh
quotients.  The min--max principle consequently gives
\begin{equation}
 \mu_{k,j}^{(r)}\leq(1+\delta_{t,r})\lambda_{k,j},
 \qquad 1\leq j\leq d_k.
 \label{eq:decorated-minmax-comparison}
\end{equation}
With
\[
 s:=\frac{t}{1+\delta_{t,r}},
\]
equation \eqref{eq:decorated-minmax-comparison} implies
\begin{equation}
 \Tr_{\mathcal M_{L,k}^{(r)}}
 \1_{[0,e_{0,r}]}(H_{L,k}^{(r),\HW})e^{-tH_{L,k}^{(r),\HW}}
 \leq
 \Tr_{\Sym^k\ell^2(\mathcal V_{L,r})}e^{-sL_{0,L,k}^{(r)}}.
 \label{eq:decorated-canonical-trace-comparison}
\end{equation}
For \(k=0\), both operators in
\eqref{eq:decorated-canonical-trace-comparison} act on a one-dimensional
space and vanish.  Hence both traces equal \(1\), so the same inequality
holds also in this sector.

It remains to sum the free canonical traces while treating the zero
eigenvalue separately.  The connected decorated graph has exactly one
one-particle zero-energy eigenfunction: the constant base function
tensored with \(u_r\).  Let
\(Z_j^{+,r}(s)\) be the free canonical trace with \(j\) particles in the
positive one-particle spectral subspace of \(h_r\), and let
\(Z_k^{0,r}(s)\) be the unrestricted free canonical trace at particle
number \(k\).  Filling the unique zero mode with the remaining particles
gives
\[
 Z_k^{0,r}(s)=\sum_{j=0}^kZ_j^{+,r}(s).
\]
It follows that
\begin{align}
 \sum_{k=0}^{K_r(t)}Z_k^{0,r}(s)
 &\leq(K_r(t)+1)\sum_{j=0}^{\infty}Z_j^{+,r}(s)\notag\\
 &=(K_r(t)+1)\mathcal G_{L,r}(s),
 \label{eq:decorated-zero-mode-summation}
\end{align}
where
\begin{align}
 \mathcal G_{L,r}(s)
 &:=\prod_{\lambda\in\spec(h_r)\setminus\{0\}}
 (1-e^{-s\lambda})^{-1}\notag\\
 &=\exp\left\{-\Tr_{P,\ne0}\log(1-e^{-sh_r})
 -\Tr_Q\log(1-e^{-sh_r})\right\}.
 \label{eq:decorated-positive-grand-canonical}
\end{align}
Eigenvalues in the product are counted with multiplicity.  Here
\(\Tr_{P,\ne0}\) is the trace on
\[
 (\ell^2(B_L)\ominus\operatorname{span}\{\1\})\otimes P\mathbb C^r,
\]
and \(\Tr_Q\) is the trace on
\(\ell^2(B_L)\otimes Q\mathbb C^r\).

Combining \eqref{eq:physical-below-decorated-trace},
\eqref{eq:decorated-su2-trace-decomposition},
\eqref{eq:quantitative-high-energy-tail},
\eqref{eq:decorated-canonical-trace-comparison}, and
\eqref{eq:decorated-zero-mode-summation}, we obtain
\[
 Z_L^{(S)}(\beta)
 \leq(M+1)\left[(K_r(t)+1)\mathcal G_{L,r}(s)
 +e^{-c_rg(\log t)^2}\right].
\]
Since \((K_r(t)+1)\mathcal G_{L,r}(s)\geq1\), the inequality
\(\log(A+B)\leq\log A+B/A\), for \(A>0\) and \(B\geq0\), yields
\begin{align}
 \log Z_L^{(S)}(\beta)
 &\leq\log(M+1)+\log(K_r(t)+1)\notag\\
 &\quad-\Tr_{P,\ne0}\log(1-e^{-sh_r})
 -\Tr_Q\log(1-e^{-sh_r})
 +C_re^{-c_rg(\log t)^2}.
 \label{eq:decorated-free-pressure}
\end{align}

We finally estimate the two invariant subspaces in
\eqref{eq:acoustic-optical-split}.  Since
\(\delta_{t,r}=o(1)\), one has \(t/2\leq s\leq t\) for all sufficiently
large \(t\).  Moreover, \(L^2=V\asymp tg\) and \(g=t^a\to\infty\), so
\(1\leq s\leq L^2\).  The restriction to the \(P\)-subspace is the scalar
Neumann Laplacian.
Lemma~\ref{lem:neumann-heat-trace} therefore gives
\begin{align}
 -\Tr_{P,\ne0}\log(1-e^{-sh_r})
 &=\frac{\pi V}{24s}
 +O\left(\frac{L}{\sqrt s}
 +\log\!\left(2+\frac Vs\right)+\frac{V}{s^2}+1\right).
 \label{eq:decorated-acoustic-pressure}
\end{align}
On the \(Q\)-subspace, \eqref{eq:optical-gap-two} and
\(-\log(1-e^{-x})\leq Ce^{-x}\) for \(x\geq2\) imply
\begin{equation}
 0\leq-\Tr_Q\log(1-e^{-sh_r})\leq C_rVe^{-2s}.
 \label{eq:decorated-optical-pressure}
\end{equation}
Furthermore,
\[
 \log(M+1)+\log(K_r(t)+1)\leq C_r\log t.
\]
Multiplying \eqref{eq:decorated-free-pressure} by \(t/V\), using
\(t/s=1+\delta_{t,r}\), and inserting
\eqref{eq:decorated-acoustic-pressure} and
\eqref{eq:decorated-optical-pressure}, we find
\begin{align}
 \frac{t}{V}\log Z_L^{(S)}(\beta)
 \leq\frac{\pi}{24}
 +C_r\bigg(&\delta_{t,r}+g^{-1/2}
 +\frac{\log t}{g}+\frac1t\notag\\[-2mm]
 &\quad+te^{-c_rt}+te^{-c_rg(\log t)^2}\bigg).
 \label{eq:quantitative-box-pressure}
\end{align}
Indeed,
\[
 \frac{t}{V}\frac{L}{\sqrt s}=O(g^{-1/2}),\qquad
 \frac{t}{V}\log\!\left(2+\frac Vs\right)
 =O\!\left(\frac{\log t}{g}\right),\qquad
 \frac{t}{s^2}=O(t^{-1}).
\]
The remaining constant heat-trace error is \(O(t/V)=O(g^{-1})\), which
is absorbed by \(O(g^{-1/2})\).  The last two terms in
\eqref{eq:quantitative-box-pressure} bound the \(Q\)-subspace contribution
and the high-energy tail, respectively.

To pass to infinite volume, tile a free square of side \(nL\) by
\(n^2\) translates of \(B_L\).  The interbox bond operators
\(S^2-\Sp_x\cdot\Sp_y\) are nonnegative, and hence the full Hamiltonian is
bounded below by the sum of the \(n^2\) box Hamiltonians.  Eigenvalue
monotonicity and factorization of the decoupled trace give, equivalently by
the Duhamel interpolation \eqref{eq:bond-dropping-duhamel},
\[
 \log Z_{nL}^{(S)}(\beta)\leq n^2\log Z_L^{(S)}(\beta).
\]
Divide by \((nL)^2\), let \(n\to\infty\), and then use
\eqref{eq:quantitative-box-pressure}.  Since
\(p_2(\beta,S)=-\beta f_2(\beta,S)\), we obtain, for every
\(0<a<1/13\),
\begin{equation}
 \beta^2S f_2(\beta,S)
 \geq-\frac{\pi}{24}
 -C_S\left(t^{13a-1}(\log t)^{16}
              +t^{-a/2}+t^{-a}\log t+t^{-1}\right)
 \label{eq:quantitative-new-direction-general-a}
\end{equation}
for all sufficiently large \(t\).  The exponentially small terms have
been absorbed into the \(C_St^{-1}\) term in
\eqref{eq:quantitative-new-direction-general-a}.

Among choices of the form \(g=t^a\), the polynomial powers are optimized
by balancing the collision error with the Neumann boundary error:
\[
 13a-1=-\frac a2,
 \qquad\text{hence}\qquad a=\frac2{27}.
\]
Thus
\begin{equation}
 \beta^2S f_2(\beta,S)
 \geq-\frac{\pi}{24}
 -C_S t^{-1/27}(\log t)^{16}.
 \label{eq:quantitative-new-direction}
\end{equation}

For the reverse direction, the two-dimensional upper bound of
\cite{NapiorkowskiSeiringer2021} gives, in multiplicative form,
\[
 f_2(\beta,S)\leq-\frac{\pi}{24S\beta^2}
 \left[1-O_S\!\left(t^{-1/3}(\log t)^{2/3}\right)\right].
\]
Equivalently, for each fixed \(S\), after enlarging an \(S\)-dependent
constant and taking \(t\) sufficiently large,
\begin{equation}
 \beta^2S f_2(\beta,S)
 \leq-\frac{\pi}{24}
 +C_S t^{-1/3}(\log t)^{2/3}.
 \label{eq:quantitative-known-direction}
\end{equation}
The error in \eqref{eq:quantitative-known-direction} is
\(o(t^{-1/27}(\log t)^{16})\).  Hence combining
\eqref{eq:quantitative-new-direction} and
\eqref{eq:quantitative-known-direction} proves
Corollary~\ref{cor:quantitative-free-energy}, and in particular completes
the proof of Theorem~\ref{thm:main-free-energy}.
\end{proof}


\begin{thebibliography}{99}
\fontsize{9.5pt}{10.1pt}\selectfont
\setlength{\itemsep}{0em}
\setlength{\parsep}{0pt}

\bibitem{Benedikter2017}
N.~Benedikter,
\newblock Interaction corrections to spin-wave theory in the large-\(S\)
limit of the quantum Heisenberg ferromagnet,
\newblock \emph{Mathematical Physics, Analysis and Geometry} \textbf{20}
(2017), no.~2, article no.~5;
\newblock \href{https://doi.org/10.1007/s11040-016-9237-6}{doi}.

\bibitem{Bloch1930}
F.~Bloch,
\newblock Zur Theorie des Ferromagnetismus,
\newblock \emph{Zeitschrift f\"ur Physik} \textbf{61} (1930), no.~3--4,
206--219;
\newblock \href{https://doi.org/10.1007/BF01339661}{doi}.

\bibitem{CaputoLiggettRichthammer2010}
P.~Caputo, T.~M. Liggett and T.~Richthammer,
\newblock Proof of Aldous' spectral gap conjecture,
\newblock \emph{Journal of the American Mathematical Society} \textbf{23}
(2010), no.~3, 831--851;
\newblock \href{https://doi.org/10.1090/S0894-0347-10-00659-4}{doi}.

\bibitem{Chen2017}
J.~P. Chen,
\newblock The moving particle lemma for the exclusion process on a weighted
graph,
\newblock \emph{Electronic Communications in Probability} \textbf{22}
(2017), paper no.~47, 1--13;
\newblock \href{https://doi.org/10.1214/17-ECP82}{doi}.

\bibitem{ConlonSolovej1991}
J.~G. Conlon and J.~P. Solovej,
\newblock Upper bound on the free energy of the spin \(1/2\) Heisenberg
ferromagnet,
\newblock \emph{Letters in Mathematical Physics} \textbf{23} (1991),
223--231;
\newblock \href{https://doi.org/10.1007/BF01885500}{doi}.

\bibitem{CorreggiGiuliani2012}
M.~Correggi and A.~Giuliani,
\newblock The free energy of the quantum Heisenberg ferromagnet at large
spin,
\newblock \emph{Journal of Statistical Physics} \textbf{149} (2012),
234--245;
\newblock \href{https://doi.org/10.1007/s10955-012-0589-4}{doi}.

\bibitem{CorreggiGiulianiSeiringer2014}
M.~Correggi, A.~Giuliani and R.~Seiringer,
\newblock Validity of spin-wave theory for the quantum Heisenberg model,
\newblock \emph{EPL} \textbf{108} (2014), no.~2, article no.~20003;
\newblock
\href{https://doi.org/10.1209/0295-5075/108/20003}{doi}.

\bibitem{CorreggiGiulianiSeiringer2015}
M.~Correggi, A.~Giuliani and R.~Seiringer,
\newblock Validity of the spin-wave approximation for the free energy of
the Heisenberg ferromagnet,
\newblock \emph{Communications in Mathematical Physics} \textbf{339}
(2015), 279--307;
\newblock \href{https://doi.org/10.1007/s00220-015-2402-0}{doi}.

\bibitem{Dyson1956}
F.~J. Dyson,
\newblock General theory of spin-wave interactions,
\newblock \emph{Physical Review} \textbf{102} (1956), 1217--1230;
\newblock \href{https://doi.org/10.1103/PhysRev.102.1217}{doi}.

\bibitem{Dyson1956Thermodynamics}
F.~J. Dyson,
\newblock Thermodynamic behavior of an ideal ferromagnet,
\newblock \emph{Physical Review} \textbf{102} (1956), 1230--1244;
\newblock \href{https://doi.org/10.1103/PhysRev.102.1230}{doi}.

\bibitem{HolsteinPrimakoff1940}
T.~Holstein and H.~Primakoff,
\newblock Field dependence of the intrinsic domain magnetization of a
ferromagnet,
\newblock \emph{Physical Review} \textbf{58} (1940), 1098--1113;
\newblock \href{https://doi.org/10.1103/PhysRev.58.1098}{doi}.

\bibitem{MerminWagner1966}
N.~D. Mermin and H.~Wagner,
\newblock Absence of ferromagnetism or antiferromagnetism in one- or
two-dimensional isotropic Heisenberg models,
\newblock \emph{Physical Review Letters} \textbf{17} (1966), 1133--1136;
\newblock \href{https://doi.org/10.1103/PhysRevLett.17.1133}{doi}.

\bibitem{NapiorkowskiSeiringer2021}
M.~Napi\'orkowski and R.~Seiringer,
\newblock Free energy asymptotics of the quantum Heisenberg spin chain,
\newblock \emph{Letters in Mathematical Physics} \textbf{111} (2021),
article no.~31;
\newblock \href{https://doi.org/10.1007/s11005-021-01375-4}{doi}.

\bibitem{SeiringerOberwolfach2025}
R.~Seiringer,
\newblock The Heisenberg ferromagnet: a dilute Bose gas in disguise,
\newblock in \emph{Mini-Workshop: New Directions in Correlated Quantum
Systems},
\newblock \emph{Oberwolfach Reports} \textbf{22} (2025), no.~1,
357--358;
\newblock
\href{https://doi.org/10.4171/OWR/2025/7}{doi}.

\bibitem{Takahashi1986}
M.~Takahashi,
\newblock Quantum Heisenberg ferromagnets in one and two dimensions at low
temperature,
\newblock \emph{Progress of Theoretical Physics Supplement} \textbf{87}
(1986), 233--246;
\newblock \href{https://doi.org/10.1143/PTPS.87.233}{doi}.

\bibitem{Thomas1980}
L.~E. Thomas,
\newblock Quantum Heisenberg ferromagnets and stochastic exclusion
processes,
\newblock \emph{Journal of Mathematical Physics} \textbf{21} (1980),
1921--1924;
\newblock \href{https://doi.org/10.1063/1.524610}{doi}.

\bibitem{Toth1993}
B.~T\'oth,
\newblock Improved lower bound on the thermodynamic pressure of the spin
\(1/2\) Heisenberg ferromagnet,
\newblock \emph{Letters in Mathematical Physics} \textbf{28} (1993),
75--84;
\newblock \href{https://doi.org/10.1007/BF00739568}{doi}.

\end{thebibliography}
\end{document}